\documentclass[10pt, conference, letterpaper]{IEEEtran}

\IEEEoverridecommandlockouts

\usepackage{cite}
\usepackage{booktabs}
\usepackage{comment}
\usepackage{graphicx,subfigure}
\usepackage{epstopdf}
\usepackage{boxedminipage}
\usepackage{alltt}
\usepackage{multicol}
\usepackage{amsmath,amssymb,amsfonts}
\usepackage{xcolor}
\usepackage{wrapfig}
\usepackage{mathptmx}
\usepackage{todonotes}
\usepackage{caption}
\usepackage{cuted}
\usepackage{listings}
\usepackage{algorithm}
\usepackage{amsthm}
\usepackage{algpseudocode}
\usepackage{graphicx}
\usepackage{capt-of}
\usepackage{textcomp}
\usepackage{xcolor}

\newtheorem{prop}{Proposition}
\newtheorem{lemma}{Lemma}
 
\newtheorem{definition}{Definition}

\usepackage{microtype}

\newcommand*{\argmin}{\operatornamewithlimits{argmin}\limits}

\def\BibTeX{{\rm B\kern-.05em{\sc i\kern-.025em b}\kern-.08em
    T\kern-.1667em\lower.7ex\hbox{E}\kern-.125emX}}

\usepackage{dsfont}
\newcommand{\rme}{\mathrm{e}}

\newcommand{\Ind}{\mathds{1}}
\newcommand{\Prob}{\mathrm{P}}

\newcommand{\new}{\textcolor{blue}}

\newcommand{\setN}{\mathcal{N}}
\newcommand{\Tree}{\mathbb{T}}
\newcommand{\voted}{\mathrm{v}}
\usepackage[affil-it]{authblk}

\begin{document}

\title{Carnot: A highly Scalable and Responsive BFT Consensus protocol}

\author{Mohammad M. Jalalzai\thanks{moh@status.im}}
\author{Alexander Mozeika }
\author{Marcin P. Pawlowski}
\author{Ganesh Narayanaswamy}
\affil{Status Research \& Development GmbH, Baarerstrasse 10, Zug, Switzerland}





\maketitle
\thispagestyle{plain}
\pagestyle{plain}
\begin{abstract}
We present Carnot, a leader-based Byzantine Fault Tolerant (BFT)  consensus protocol that is  responsive and operates under the partially synchronous model. 
Responsive BFT consensus protocols exhibit wire-speed operation and deliver instantaneous finality, thereby addressing a fundamental need in distributed systems. A key challenge in scaling these protocols has been the computational complexity associated with authenticator verification. We demonstrate that Carnot effectively addresses this bottleneck by adeptly streamlining the verification and aggregation of O(log(N)) authenticators per node. This notable advancement marks a substantial improvement over the prevailing O(N) state-of-the-art approaches. Leveraging this inherent property, Carnot demonstrates its capacity to seamlessly scale to networks comprising tens to hundreds of thousands of nodes. We envision Carnot as a critical stride towards bridging the gap between classical BFT consensus mechanisms and blockchain technology.

\end{abstract}


\section{Introduction}
Recently there has been a lot of interest in using BFT-based consensus protocols in blockchains \cite{Hot-stuff, Algorand, Hentschel2020FlowSC, libra-BFT, Casper,ethereum}. The main reason for this attention  is the higher performance of these protocols,  i.e. high throughput, low latency, and the absence of a need for expensive mining to reach a consensus on a value. The most recent noticeable event related to the latter was the Ethereum's merge \cite{ethereumMerge}, where Ethereum has merged its main chain, which used PoW consensus, into the PoS chain that runs a BFT-based consensus thereby reducing the cost of gas by  $99.5 \%$. 

BFT is the ability of a computer system to tolerate   arbitrarily, i.e. Byzantine\cite{Lamport:1982:BGP:357172.357176},   faults of its components during the operation. State machine replication  (SMR) refers to the ability of a system to replicate its state across $N$ nodes in a deterministic manner. Hence, the BFT SMR system will provide SMR services despite  efforts of (at most $M$) Byzantine nodes to break the system. As is common, we are interested in a partially synchronous communication model \cite{Dwork:1988:CPP:42282.42283} where there is a known maximum bound on message delivery after some unknown \textit{global stabilization time} (GST).

BFT-based consensus algorithms can be divided into two categories. The first category  comprises of  responsive, but not scalable  algorithms\cite{Fast-HotStuff,2018hotstuff,Castro:1999:PBF:296806.296824,SBFT}, while the  second category is comprised of scalable  but not responsive algorithms which also, suffer from chain reorganization, and the lack of instant finality \cite{Ethereum-Gasper,Casper}. Therefore, there is a need for a consensus protocol that is responsive, has instant finality, is not prone to chain reorganization, and scales to tens of thousands and potentially hundreds of thousands of nodes. {In the next three paragraphs}, we delve into  essential features that are sought after, and highlight \ {all} repercussions that arise from their absence\new{,} in  {a} consensus protocol.

\textbf{Responsiveness.}
A distinguishing characteristic of partially synchronous BFT consensus is its  \emph{responsiveness}. The latter means that, given a non-Byzantine leader, the protocol operates at the speed of the network, rather than being bound by a maximum message delay. In other words, during normal execution, the protocol is event-driven and does not require waiting for a round completion. On one hand  on a partially synchronous communication model, optimistic responsiveness holds greater relevance as it is achieved in optimistic scenarios, typically after the global stabilization time (GST). On the other hand, protocols that lack responsiveness often introduce a block (or slot)  time parameter, which dictates that the protocol advances to the next round only at the end of the designated time interval. Unfortunately, this reliance on time not only slows down the protocol but also creates vulnerabilities to timing-based attacks. A recent example of such an attack is the one that affected the Ethereum Gasper, which exploited the protocol's time dependency \cite{Eth-PoS-Timing-Attack}, resulting in a loss of approximately $20$ million USD. Furthermore, responsiveness significantly influences another vital property of a consensus protocol known as instant-finality.

\textbf{ Finality.} In blockchains, \emph{finality} refers to the confirmation that a committed block will not be revoked. It is also called safety property in distributed systems consensus.  The time taken  from the moment a client submits a transaction to the network and then receives affirmation of its finality/safety is called \emph{latency}. Therefore, a low latency protocol exhibits the instant-finality property. Instant finality improves the client's experience and paves the way for many blockchain use cases that weren't possible due to the low latency requirements. Moreover, finality prevents frequent chain reorganization. Chain reorganization also opens attack vectors when Proof-of-Stake (PoS) is implemented \cite{TwoAttacksOnEthereum}. Changing the algorithm to address these attacks  is very  expensive and make\new{s} protocol complex \cite{EthereumSingleSlotFinality}.

\textbf{Elastic scalability.} Elastic scalability refers to a consensus protocol's ability to adapt to changes in the number of nodes in a distributed network. Carnot BFT is the first  protocol to possess this property, making it ideal for blockchain networks that expand over time. While BFT consensus was originally designed for small networks, interest in scaling has grown with the development of blockchain technology. Several scalable BFT consensus protocols have been proposed, including HotStuff \cite{2018hotstuff}, Fast-HotStuff \cite{Fast-HotStuff}, Hermes \cite{jalalzai2020hermes}, and SBFT \cite{SBFT}, but they are limited to a few hundred nodes. The Thundrella \cite{Thundrella} and Ethereum Gasper \cite{Ethereum-Gasper} take a different approach, using BFT consensus within a small committee chosen from a larger pool of nodes. However,  the former relies on proof-of-work as a fallback solution, and the latter  has latency issues and a high risk of chain reorganization. This leads to the following question.

\textbf{
Is it possible to  build a  protocol that has  elastic scalability, is  responsive and achieves  Instant finality? 
}
In this paper, we present Carnot  as the first, to the best of our knowledge, responsive  BFT-based consensus protocol with elastic scalability. Carnot can scale to accommodate a large number of nodes whereas it can still provide optimistic responsiveness, which helps it to achieve instant finality. 
Moreover, it also prevents attacks on PoS due to chain reorganization.

Authenticator complexity is one of the main bottlenecks in the consensus protocols \cite{2018hotstuff, EthereumSingleSlotFinality}.  To avoid chain reorganization and achieve responsiveness and finality, a protocol has to make sure the majority of nodes have attested  the proposal. However,  this verification requires at least  $O(N)$ signatures to be aggregated and verified. Although using a peer-to-peer network significantly reduces the cost of sending  authenticators,  the verification cost remains a problem. 

To address this issue, responsive pipelined BFT-based protocols \cite{2018hotstuff,FastHotStuff-IEEE,No-Commit-Proofs, Marlin} have successfully reduced the authenticator verification cost to  $O(N)$. These protocols typically operate by having the leader collect votes  from the previous block and construct a quorum certificate,   $QC$, from these votes, which is then included in the current block. Upon receiving the proposal, each node needs to verify the $QC$, which serves as proof that more than $2N/3$ nodes in the network have voted. While this approach works effectively for networks with tens or a few hundred nodes, it becomes  inefficient as the network size grows to thousands or tens of thousands of nodes due to the aggregation and verification of votes from such a large number of nodes.

Given  that the authenticator verification and/or aggregation  cost, while maintaining responsiveness and finality, is the  main bottleneck in scaling  consensus protocols, we provide an intuitive solution to this problem. To this end, we designed a protocol where a node can verify that  more than  $2N/3$ of nodes, where $N$ is in     thousands ( even tens {or hundreds} of thousands), have voted by verifying only a small number of authenticators. Further details  of this  protocol  will be discussed below. 

Assuming that each node in the network is assigned to a committee, with committees forming  binary tree overlay  (see  Figure \ref{fig:comm-tree}), and each node has successfully verified a block, proposed by the leader, the voting process begins at the  leaves of this tree. Members of leaf committees send their votes to the members of  parent committees. Then members of  a parent committee collect votes from at least two-thirds of their child committee  and forward  the votes to  their parent. This  process continues until  votes reach the root committee. Next members of the root committee send their votes, along with $QC$ representing at least two-thirds of  votes from the members of the root committee's children, to attest to the proposal \footnote{A member of the root committee will also forward any additional distinct vote it has received after casting its vote. More about this will be discussed later.}. Finally, the leader  proposes a block and includes  proof of votes from the root committee and its children in the form of $QC$ (from the parent block) appended to the current block.

The described above process of verifying  $QC$, constructed from signatures of the root committee and its child committees\footnote{The count of committees situated at the upper levels of the binary tree overlay, whose members' signatures are encompassed within the $QC$, can be parameterized to offer a customizable approach.}, ensures that the majority of nodes in the network have indeed  voted for the block. By collecting and aggregating votes from these committees, the protocol establishes strong evidence that more than two-thirds of the network participants have endorsed the proposed block. This low-cost verification step provides confidence in the validity of the block and strengthens the consensus achieved by the protocol.

The  {results of } section \ref{section:failures}  {suggest} that committee size{s} grow logarithmically  {with respect to the } total number of nodes, $N$, in the network. {The latter} allows  {for a} node {in a parent committee} to verify signatures from {its} child committees {(see Figure \ref{fig:comm-tree-connect})}  {by collecting} only $O(\log N)$ signatures. In a similar manner, the number of signatures that  need to be aggregated in the top three committees is also $O(\log N)$. 
Hence the protocol is designed to scale  {by  reducing the cost}  of {the} signature verification and aggregation. This scalability ensures  the ability {of the protocol} to uphold {the} responsiveness and finality {properties}.




\section{Related Work} \label{Section: Related Work}
Currently, there is  no other, to the best of our knowledge, consensus protocol  that combines elastic scalability, finality and responsiveness. While there are consensus protocols that exhibit subsets of these properties, achieving all three in a single protocol remains a challenge.

On one hand there are protocols  which are responsive and provide finality, such as HotStuff \cite{2018hotstuff}, Fast-HotStuff \cite{FastHotStuff-IEEE}, Wendy \cite{No-Commit-Proofs}, SBFT \cite{SBFT}, \cite{Proteus1}, Window based BFT \cite{Jalal-Window}, PBFT \cite{Castro:1999:PBF:296806.296824}, BFT Smart \cite{BFT-SMART}, Aardvark \cite{Aardvark}, \cite{Prime} and Zyzzyva \cite{Kotla:2008:ZSB:1400214.1400236}, but the complexity of their authenticators varies  from    $O(N)$ to $O(N^2)$. Another group of consensus protocols prioritizes finality  over responsiveness \cite{Casper,tendermint,POE}. The latter $O(N)$ authenticator complexity.

On the other hand, there are consensus protocols that offer scalability but compromise on finality and responsiveness. However, it should be noted that to the best of our knowledge, there is no known protocol that provides \emph{elastic} scalability. Nevertheless, there are PoS protocols that can accommodate a large number of nodes. One  such  protocol is the Gasper used in Ethereum  \cite{Ethereum-Gasper}. The latter achieves scalability by operating on a small subset of  nodes randomly selected from  the set of \emph{all} nodes. However, it relies on the maximum block generation time interval (or slot time). This time dependency introduces delays between rounds or views, which hinders the attainment of finality and increases the likelihood of forks or chain reorganizations. Consequently, the Gasper protocol provides a suboptimal user experience, becomes susceptible to timing and PoS attacks, and requires additional complexity to address these vulnerabilities.  The latter makes  the protocol more challenging to maintain and to understand  its complexities\cite{EthereumSingleSlotFinality}.

 To summarise, while there are consensus protocols that excel in certain aspects, finding a protocol that combines elastic scalability, finality, and responsiveness in an optimal manner remains a challenge. Here, however, we are proposing  a new consensus protocol, the Carnot, that aims to accommodate all of these properties within a single protocol. By leveraging innovative design principles and mechanisms, Carnot seeks to achieve elastic scalability, providing the ability to adapt to changing network sizes, while maintaining finality and responsiveness. Through extensive research and analysis, we aim to address the limitations of existing protocols and pave the way for a more efficient and robust consensus mechanism in blockchain networks.


\section{System Model and Preliminaries}\label{Section: System Model} 
\subsection{System Model}
We consider a system  of $N$  nodes  where \emph{at most}  $M<\frac{1}{3}N$ nodes may be Byzantine.  $M$ is parameterized to achieve the desired committee size and failure probability. While honest (or correct)  nodes follow the protocol, the Byzantine nodes  may deviate from the latter in an arbitrary manner. We assume a \emph{partial synchrony} model \cite{Dwork:1988:CPP:42282.42283}, i.e. there exists a known bound $\Delta$ on message transmission delay, applicable after an unknown asynchronous period called \textit{Global Stabilization Time} (GST). In practice, the system can only make progress if the $\Delta$ bound  persists for a sufficiently long time, and we simplify the discussion by assuming  that this $\Delta$ bound  is \emph{everlasting}. Furthermore, we assume that \emph{all} messages exchanged in the system are \emph{signed}, and adversaries are computationally bound and cannot forge signatures (or message digests) with more than negligible probability.  Finally, we assume that  Carnot uses   \emph{static adversary} model, i.e.  nodes ``decide'' if they are correct  or corrupted at the beginning of the protocol and do not change this label during its operation.

\subsection{Preliminaries}

\noindent \textbf{View and View Number.}
A view is identified by a monotonically increasing view number. During each view, a leader is responsible for  submitting a proposal, and a deterministic (pseudo-random) function is used to select a leader for that view. This process ensures fairness and randomness in the leader selection process. The pseudo-random function  is based on a random seed generated by a random Beacon.
 The deterministic function is used for leader selection in every view, but during a failure, the random Beacon is also used to generate a new overlay structure to maintain system availability.

\noindent 
\textbf{Quorum Certificate (QC), Aggregated QC and $\textbf{high\_qc}$.}  Quorum Certificate (QC) is a cryptographic proof that attests \ agreement (or support)  on a specific proposal (or timeout message) within a group of nodes. It is created by aggregating  signatures of participating nodes. In the Carnot consensus protocol, there are two types of latest QC, referred to as the local $high\_qc$ and the global $high\_qc$. The local $high\_qc$ represents the most recent QC that a node has observed (or received). It serves as a  reference  which is used by the node to determine the validity and ranking of blocks in its local view.  In contrast to local $high\_qc$, the global $high\_qc$ represents the most recent QC present in the Aggregated QC ($Aggregated\_{qc}$). The latter is a compressed representation of the vector that combines multiple block attestation QCs, providing a collective overview of the consensus state after unhappy path (failure).
In addition to block attestation QCs, Carnot also utilizes timeout QCs.  The latter is a QC specifically associated with timeout messages  and it serves as a proof of the agreement (or acceptance) of the timeout event within  a  group of nodes thereby providing a mechanism to handle failures in the consensus protocol.

{\textbf{Cryptographic scheme selection note.} The selection of the signature aggregation scheme for the implementation of AggQC is a crucial decision. It will directly impact the execution time of the protocol, due to the overhead of signatures aggregation and verification. In addition, the decision should take into account the expected size of a committee. In \cite{li2023performance} it was shown that EdDSA should be favored over BLS for larger deployments (more than 40 validators). The same reasoning is expected for deciding on a source of randomness for the network. A random beacon construction, in addition to the security of the scheme, will also impact the execution time of the protocol. A fully decentralized and multi-party random beacon will increase the security of the network at the expense of longer randomness generation time, especially for large-scale deployments.
Therefore, we leave these decisions open for particular instances of the protocol.}


\noindent \textbf{Block and Block Tree.}
The transactions received from clients are batched together by  the leader into a block.  The block has a field of the type $QC$. The leader adds  $QC$ of the parent-block, i.e. the preceding block,  of the block  being proposed  to this field. In this way the block points to its parent block using  the $QC$  of the parent block. We note that every block, except the genesis block, must specify its parent block and include  $QC$  of the parent block. In this way all blocks are chained together. 
As there may be forks during failure, each node maintains a  tree,  referred to as the  $blockTree$ , of received blocks. However,  blocks that are committed  are also linearly ordered and hence hold a specific position in the blockchain.


\noindent \textbf{Chain and Direct Chain.} If  block $B$ is  {stacked atop}   block $B^\prime$ (such that $B.parent=B^\prime$), then these two blocks  form the one-chain. If another block $B^*$ is {placed atop} block $B$, then $B^\prime$, $B$, and $B^*$  form the  two-chain and so on. 
Chains grow in two ways. The first is continuous growth when blocks from consecutive views are directly connected one after the other. For instance, if $B.curView = B'.curView+1$ and $B.parent=B'$, we establish a direct link between $B$ and $B'$, forming a one-direct chain. If $B.curView = B'.curView+1$, $B.parent=B'$, and $B^*.curView = B.curView+1$ with $B^*.parent=B$, a two-direct link is formed among $B'$, $B$, and $B^*$.

However, sometimes blocks might not be generated due to leader or network issues. In these cases, the chain can jump ahead. For example, if $B.curView = B'.curView+k$ ($k>1$) with $B.parent=B'$, there's no direct link between $B$ and $B'$, reflecting a gap in the chain.

\section{Overlay Tree Formation and Message Dissemination }
\label{Section: Overlay-Formation-And-Msg-Dessimination}
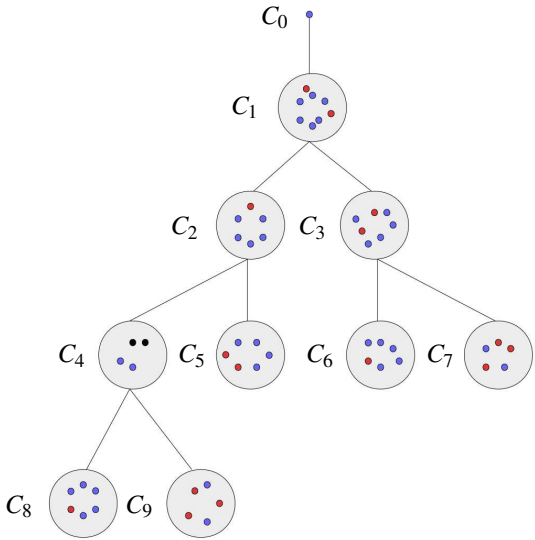
\begin{figure}[!t]
\setlength{\unitlength}{1mm}
\begin{picture}(100,75)
\put(7,0){\includegraphics[height=75\unitlength]{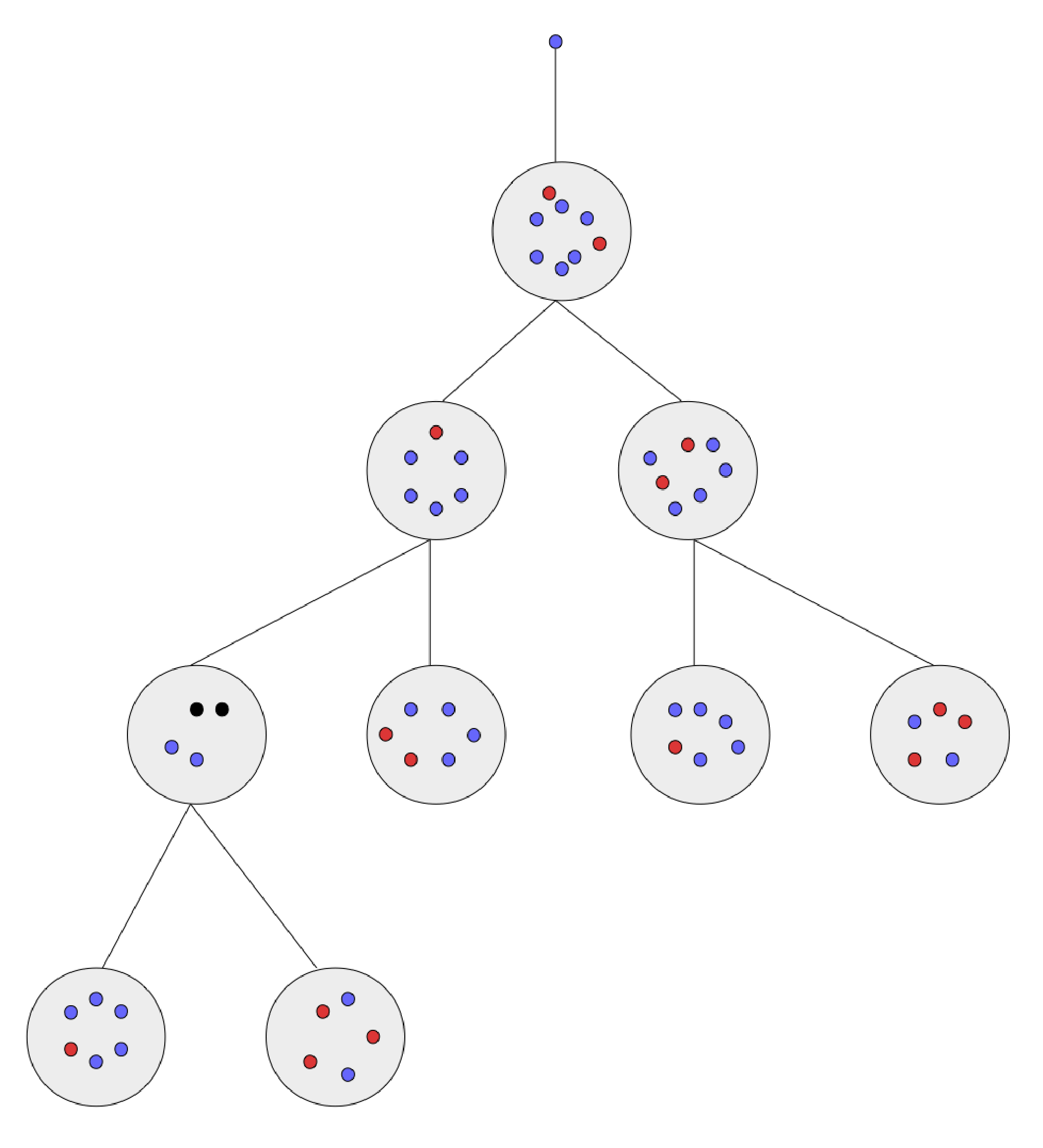}}
\put(37,71){$C_0$}
\put(33,59){$C_1$}
\put(25,43){$C_2$} \put(42,43){$C_3$} 
\put(10,26){$C_4$} \put(26,26){$C_5$} \put(43,26){$C_6$}  \put(59,26){$C_7$} 
\put(3,6){$C_8$}  \put(19,6){$C_9$} 
\end{picture}
\caption{The $N$ nodes in the set $\setN$, represented by dots,   are assigned in a \emph{random} and \emph{unbiased} way  into the \emph{odd} number, $K$,  of committees $C_\mu$, represented by circles, in a such a way that $\setN=\bigcup_{\mu=1}^{K} C_{\mu}$. The \emph{leader} {, represented by the committee $C_0\in \setN$,} is a special node connected to the root committee $C_1$. 
}
\label{fig:comm-tree} 
\end{figure}
%

The Carnot presented here is a highly scalable variant of the family of responsive pipelined consensus protocols~\cite{Hot-stuff,Fast-HotStuff, No-Commit-Proofs}.  It uses the underlying peer-to-peer network layer for communication among nodes 
  , but in order  to scale, we use an  overlay tree of committees (see Figure \ref{fig:comm-tree}) for  voting (or $new\_view$\footnote{$new\_view$ message is used by nodes in the network to enter a new view when failure is detected. More details are given in subsection \ref{Subsection: Detailed Description}} message collection). The tree of committees is defined as follows
   \begin{definition}[Committee Tree]
   \label{def:comm-tree}
  For the set of  nodes $\setN=\bigcup_{\mu=1}^{K} C_{\mu}$, where $N=\vert \setN\vert$, the committee tree $\Tree$   is a connected, undirected, and acyclic graph $\mathbb{T} = (V, E)$, where $V=\{C_{0},C_{1},\ldots,C_{K}\}$ is the set of vertices and $E$,  {with} $\vert E\vert=K$, is the set of edges. Here the committee $C_0$ contains exactly one node, i.e.  $\vert C_0\vert=1$, and $C_0\in \setN$.  
\end{definition}
  In this overlay structure,  which is independent of the network layer each  vertex in the tree is a committee of 
nodes.  The latter is used, instead of a single node,  to improve robustness. 
\textbf{The novel quorum system of the Carnot} helps the protocol  to scale while maintaining  responsiveness and instant finality. However, this design choice results in a probabilistic nature of the network  and of the consensus layer. A more detailed discussion of this will be presented later, but now we define  properties that must be satisfied by the tree  of committees and peer-to-peer network.

\begin{definition}[Robust Child Committees]
\label{def:robust-child-comm}
The child committees $C_{\nu}$ and $C_{\mu}$ which share the same parent in the committee tree $\Tree$  are robust if at least $\frac{2(N_{\nu}+N_{\mu})}{3}$ nodes, where  $N_{\nu}=\vert C_{\nu}\vert$ and $N_{\mu}=\vert C_{\mu}\vert$, are correct in these committees. 
\end{definition}


\begin{definition}[Probabilistic Reliable Dissemination]
  After the GST, and when the leader is correct and the peer-to-peer network is robust, all the correct nodes deliver the proposal sent by the leader with high probability (w.h.p).

\end{definition}


For liveness, we assume that  peer-to-peer network  {underlying the} Carnot  satisfies  {properties of} the probabilistic reliable dissemination. In practice,  reliability of the  peer-to-peer networks will depend on the number of peers  a node is forwarding the received message to. The  details of how the data is received and reshared in the peer-to-peer overlay network are provided in \cite{WakuSpec,WakuPaper}.
 
\subsection{Overlay Formation}
The formation of the overlay structure (see Figure \ref{fig:comm-tree}) takes a set of nodes $\setN$, size of the committee $n$, that is calculated for a \emph{given} probability of failure $\delta$ as described in the Algorithm \ref{alg:comm}, and a random input $\xi$. We assume that values of these parameters  are the same for all nodes participating in the overlay formation process. The latter is straightforward, as presented in the Algorithm~\ref{alg:form-tree}, and can be described as follows. First, the set of nodes $\setN$ is permuted with the random input $\xi$. This procedure must be deterministic and all permutations must be equally likely. To this end we are using the Fisher-Yates shuffle as described in \cite{knuth1997art}. 

Next, using parameter $n$, we would like to divide the set $\setN$ into a set of committees of size $n$. In particular $\setN = \bigcup_{\mu=1}^{K} C_{\mu}$, where $C_{\mu} \subseteq\setN$ is the committee $\mu$, and $K$ is the number of committees. We note that $  C_{\mu}\bigcap C_{\nu}=\emptyset$ for all  $\mu\neq\nu$. The number of committees $K = \lfloor N / n \rfloor$, but $n$ might not split the set $\setN$ into subsets of equal sizes and a reminder (of nodes) is expected. To mitigate this we are adding an extra node to the last $r$ committees, where $r$ is the reminder, i.e. each committee indexed by $\mu \in <K-r, K>$ has $n+1$ nodes, while the rest have $n$ nodes in each committee. We note that by  having (almost) equal committee sizes we are ensuring that likelihood of one type of failure (see Figure \ref{fig:failures}), which depends on committee sizes (see section \ref{section:failures}),  is not decreased be increasing likelihoods of other type of failures. Also  having   uneven  committee sizes can increase the network latency, e.g. a root committee with $n+r$ nodes would require more resources, depending on the $n$ to $r$ ratio, to move forward the execution of the consensus.    

Finally, the set of committees $\{C_1, \ldots, C_K\}$ is mapped onto a binary tree where the index $\mu$ of the committee $C_{\mu}$ mapped onto one of its  vertices. The latter is implemented via a simple scheme where the parent vertex, labelled by $\mu\new{,}$ has the left and right child vertices labelled, respectively, by $2\mu$ and $2\mu+1$ Thus, the root committee is $C_{1}$, its left child is $C_{2}$ and its right child is $C_{3}$. Children of $C_{2}$ are formed with $C_{4}$ and $C_{5}$, and so on (see Figure \ref{fig:comm-tree}). 

We note that the generation of the overlay tree and node assignment to committees is deterministic. Therefore, verification and interpretation of the overlay construction and node membership does not require additional synchronization and is given by common view of the algorithm inputs. 

\begin{algorithm}[tb]
\caption{Overlay Tree Formation}
\label{alg:form-tree}
\begin{algorithmic}[1]
\small
\Require $\setN$ is set of nodes, $n$ is size of the committee, $\xi$ is random value
\Ensure an overlay tree is formed
\State $K \gets \lfloor N / n \rfloor$ \Comment{$K$ is the number of committees}
\State $r \gets N \mod n$ \Comment{$r$ is a reminder of nodes}
\State $nodes \gets \Call{Shuffle}{\setN, \xi}$ \Comment{Set of nodes $\setN$ is shuffled with a random seed $\xi$}
\For{$K >= \mu >= 1$}
\If{$r > 0$}
    \State $C_{\mu} \gets \Call{Assign}{nodes, n+1}$ \Comment{Assign $n+1$ nodes to $r$ last committees}
    \State $r \gets r - 1$
\Else
    \State $C_{\mu} \gets \Call{Assign}{nodes, n}$ \Comment{Assign $n$ nodes to the rest of the committees}
\EndIf

\State $V = \{C_1, \ldots, C_K\}$

\EndFor
\For{$1 <= \mu <= K$}
    \State $E \gets \Call{AddLeftChildEdgeIfExist}{C_{\mu}, C_{2\mu}}$ \Comment{Add a left child edge from $C_{\mu}$ to $C_{2\mu}$ if $C_{2\mu}$ exists}
    \State $E \gets \Call{AddRightChildEdgeIfExist}{C_{\mu}, C_{2\mu+1}}$ \Comment{Add a right child edge from $C_{\mu}$ to $C_{2\mu+1}$ if $C_{2\mu+1}$ exists}
\EndFor

\State \Return $\mathbb{T} = (V, E)$ \Comment{Return the overlay tree}
\end{algorithmic}
\end{algorithm}


\begin{algorithm}[tb]
\caption{Propose Block}
\label{alg:propose-block}
\begin{algorithmic}[1]
\small
\Require $view$: current view, $quorum$: set of messages from the view
\Ensure broadcast the proposed block to the network
\State $\text{assert } \textit{ is\_leader(id)}$
\State $qc \gets \text{null}$
\State $quorum \gets \text{list}(quorum)$
\If {$quorum[0]$ is  of  type (Vote)}
    \State $\text{assert } \textit{len(quorum) } \geq \textit{ leader\_super\_majority\_threshold()}$
    \State $vote \gets quorum[0]$
    \State $qc \gets \textit{build\_qc}(vote.view, \textit{safe\_blocks}[vote.block], \text{null})$
\ElsIf {$quorum[0]$ is  of type($new\_view$)}
    \State $\text{assert } \textit{len(quorum) } \geq \textit{ leader\_super\_majority\_threshold()}$
    \State $new\_view \gets quorum[0]$
    \State $Aggqc \gets \textit{build\_Aggqc}(new\_view.view, \text{null}, quorum)$
\EndIf
\State $block \gets \text{Block}(view,qc,\textit{[]txs})$
\State $\Call{broadcast}{block}$
\end{algorithmic}
\end{algorithm}

\begin{algorithm}
\caption{Approve Block}
\label{alg:approve-block}
\begin{algorithmic}[1]
 \small   
\Require $block$, $votes$
\Ensure send vote, increment voted view and view number.
    \State \textbf{assert} $block.id() \in$ \texttt{safe\_blocks}
    \State \textbf{assert} $\operatorname{length}(votes) = $ \texttt{ super\_majority\_threshold}( )
    \State \textbf{assert all} \texttt{ is\_member\_of\_child\_committee}( $vote.voter$) \textbf{for} $vote$ \textbf{in} $votes$
    \State \textbf{assert all} $vote.block = block.id()$ \textbf{for} $vote$ \textbf{in} $votes$
    \State \textbf{assert} $block.view > \text{ highest\_voted\_view}$
    \If{\texttt{ is\_member\_of\_root\_committee}( )}
        \State $qc = \text{ $build\_qc$}(block.view, block, \text{votes})$
    \Else
        \State $qc = \text{None}$
    \EndIf
    \State $vote: \text{Vote} = \text{Vote}(block=block.id(), voter=\text{self.id}, view=block.view, qc=qc)$
    \If{\texttt{ is\_member\_of\_root\_committee}( )}
        \State \textbf{send}($vote$, \texttt{ leader}(block.view + 1))
    \Else
        \State \textbf{send}($vote$, \texttt{ parent\_committee}( ))
    \EndIf
    \State $\text{ increment\_voted\_view}(block.view)$ \Comment{to avoid voting again for this view.}
    \State $\text{ increment\_view\_qc}(block.qc)$
\end{algorithmic}
\end{algorithm}

\begin{algorithm}
\caption{Timeout Detected}
\label{alg:timeout-detected}
\begin{algorithmic}[1]
\Require{$msgs$: is a set of timeout messages received by the root committee}
\Ensure{The root committee broadcasts a timeout QC message to all nodes, indicating that a new view needs to be created}
\State \Comment{Root committee detected that supermajority of root + its children has timed out}
\State \textbf{assert} $len(msgs) == leader\_super\_majority\_threshold()$
\State \textbf{assert} $all(msg.view \geq current\_view \text{ for } msg \text{ in } msgs)$
\State \textbf{assert} $len(set(msg.view \text{ for } msg \text{ in } msgs)) == 1$
\State \textbf{assert} $is\_member\_of\_root\_committee(self.id)$
\State $timeout\_qc \gets 
\Call{build\_timeout\_qc}{msgs, self.id}$
\State \Call{broadcast}{$timeout\_qc$} \Comment{we broadcast so all nodes can get ready for voting on a new view}
\State \Comment{Note that $receive\_timeout\_qc$ should be called for root nodes as well}
\end{algorithmic}
\end{algorithm}

\begin{algorithm}
\caption{Receive Timeout QC}
\label{alg:receive-timeout-qc}
\begin{algorithmic}[1]
\Require{$timeout\_qc$ is a timeout QC message received by a node}
\Ensure{The node updates its state based on the received timeout QC message and proceeds to the next step in the protocol}
\State \textbf{assert} $timeout\_qc.view \geq current\_view$
\State $new\_high\_qc \gets timeout\_qc.high\_qc$
\State \Call{$update\_high\_qc$}{$new\_high\_qc$}
\State \Call{$update\_timeout\_qc$}{$timeout\_qc$}
\State \Comment{Update our current view and go ahead with the next step}
\State \Call{$update\_current\_view\_from\_timeout\_qc$}{$timeout\_qc$}
\State \Call{$rebuild\_overlay\_from\_timeout\_qc$}{$timeout\_qc$}
\end{algorithmic}    
\end{algorithm}

\begin{algorithm} 
\caption{Approve $new\_view$}
\label{alg: approve-newview}
\begin{algorithmic}[1]
\Require{$timeout\_qc$ is a timeout QC message and $new\_views$ is a set of new view messages received by a node}
\Ensure{The node approves the new view, updates its state, and broadcasts the new view message to relevant nodes}
\If{$last\_view\_timeout\_qc \neq None$}
\State \textbf{assert} $\forall new\_view \in new\_views$, $new\_view.view > last\_view\_timeout\_qc.view$
\EndIf
\State \textbf{assert} $\forall$ $new\_view \in new\_views$, $new\_view.timeout\_qc.view = timeout\_qc.view$
\State \textbf{assert} $len(new\_views) = overlay.super\_majority\_threshold()$
\State \textbf{assert} $\forall$ $new\_view \in new\_views,$\
\hspace{1cm} $is\_member\_of\_child\_committee( new\_view.sender)$
\State \Comment{The new view should be for the view successive to the timeout}
\State \textbf{assert} $\forall$ $new\_view \in new\_views, timeout\_qc.view + 1 = new\_view.view$
\State $view \gets timeout\_qc.view + 1$
\State \textbf{assert} $highest\_voted\_view < view$

    \State \Comment{Get the highest QC from the new views}
    \State $messages\_high\_qc \gets \{new\_view.high\_qc \mid new\_view \in new\_views\}$
    \State $high\_qc \gets \max([timeout\_qc.high\_qc, *messages\_high\_qc], key=\lambda qc: qc.view)$
    \State \Call{update\_high\_qc}{$high\_qc$}
    \State $timeout\_msg \gets$ $new\_view$(view=view, high\_qc =local\_high\_qc , sender=id, timeout\_qc=timeout\_qc)

    \If{is\_member\_of\_root\_committee(self.id)}
        \State $AggregatedQC \gets buildAggQC(new\_views)$
        \State $timeout\_msg.Aggregated\_qc = AggregatedQC$
        \State \Call{send}{$timeout\_msg$, overlay.leader(current\_view + 1)}
    \Else
        \State \Call{send}{$timeout\_msg$, *overlay.parent\_committee(id)}
    \EndIf

    \State \Comment{This checks if a node has already incremented its voted view by local\_timeout. If not then it should do it now to avoid voting in this view.}
    \State highest\_voted\_view $\gets \max\{\text{highest\_voted\_view}, view\}$
\end{algorithmic}
\end{algorithm}

\section{{The} Carnot {protocol}}
\label{Section: Carnot}
{The protocol}  is designed to operate in rounds, where each round corresponds to a view in the system. At the beginning of each round, a leader is chosen for  {this} round, and the leader proposes a batch of transactions (in the form of a block) that the nodes will attempt to agree on.  First, we will provide an overview of   message types {used in Carnot} followed by  a description of the protocol.

\subsection{Message Types}
When node casts a vote 
for a proposed block, it creates {the \emph{Vote}} object {$\langle view, block.Id(), voter \rangle$} that includes the view number, the hash value of the proposed block {,}  computed {by the}  $block.Id()$ {function}, and  identifier of the voting node. The vote is then signed over the $view$ and $block.Id()$ fields.

{The} quorum certificate (QC), represented by $\langle view, block.Id(), voters \rangle$,  is a collection of votes for the same view and block hash ($block.Id()$), along with aggregated signatures of the voters over the view and $block.Id()$ fields. The set of voters in the QC is represented as a bit array with a single $1$ bit in the same index as the corresponding node identifier in the array of node IDs.

 {The} timeout message  {, given by} $\langle view, high\_qc, sender \rangle$ is used to inform other nodes {of a} {about the} timeout event.  {The} timeout  {QC} $timeout\_qc$ is a  collection of timeout messages, signed over the view of the timeout message and $high_{qc}.view$.

The new view message, given by   $\langle view, high\_qc, sender, timeout\_qc, {Aggregated\_qc} \rangle${,} is signed over the current view and the view of the $high\_qc$. {$Aggregated\_qc$ is only used by the root committee to forward the $Aggregated\_qc$ of its children committees to the next leader, similar to  forwarding votes in happy path\footnote{Alternatively, the members of the children committees of the root committee broadcast their $new\_view$ messages to the root committee and the next leader. Each member of the root committee also sends their $new\_view$ message to the next leader upon receipt of the $super\_majority\_threshold$ of the $new\_view$ messages from its children committee members.}.} 

 {The} $Aggregated\_{qc}$ $\langle view, qc\_views, senders, high\_qc \rangle$  {is} a collection of $new\_view$ messages. {The next leader builds the aggregated signature of the  aggregated $QC$} by aggregating the signatures of $new\_view$ messages. 

After receiving  the proposal the following steps occur\new{:}
\subsection{ {The} Happy Path Overview}
\begin{itemize}
    \item Each node in the network independently verifies the validity of a new block.

    \item If there is a parent committee, nodes in the leaf  committee send their votes to the {nodes}  {in} the parent committee.
    \item  {Each node} {in the} parent committee waits to receive  two-thirds of the votes from both child committees before voting on the block's validity themselves {(see Figure \ref{fig:comm-tree-connect})}. Once they have reached the required threshold, each node from the parent committee casts their vote on the block's validity to their parent committee(if there is any).
    \item Nodes at  the root committee cast their vote for the leader of the next view. If the root committee has children committees, then it will also forward the quorum certificate of children committees in its vote to the next leader.
    \item Once, a member of the root committee has cast its vote, it forwards any additional distinct vote to the leader of the next view.
    \item The leader of the next view builds a QC from more than two-thirds of the votes (from the root committee and its children combined) it received from the root committee. It adds the QC to the block and proposes it to the network.
\end{itemize}

\subsection{ {The} Unhappy Path Overview}
{The protocol}  switches to an unhappy path when a failure is detected. {The unhappy and happy paths are very similar. A}  failure is detected by  {inspecting the }  $timeout\_qc$ {for the current view or a higher view. We note that} this step is analogous to the block proposal in the happy path. Whereas, {the} $new\_view$ message moves through the overlay similar to the votes in the happy path. This upward movement of  $new\_view$ messages  {ensures that} nodes agree to begin the next view and also {to} move the most recent $QC$ (global $high\_qc$) upward. The  {latter} will be used by the next leader to propose the next block pointing to the $high\_qc.block$ as its parent.
More details about the unhappy path are given in the next subsection  {,} but here we provide a brief overview.
\begin{itemize}
    \item Upon timeout, all  nodes  including the members of root and its children committee{s,} stop participating in the consensus.
    \item Member of the root committee and its children send their timeout messages to the root committee.
    \item {The nodes in the} root committee  build a $timeout\_qc$ and broadcast it to the network.
    \item Each node recalculates the new overlay.
    \item {The} nodes in  leaf {committees} of the new overlay send their $new\_view$ message{s} to the  {nodes} of their parent committee.
     \item Upon receipt of   two-thirds of $new\_view$ messages from child committees, each node in the parent committee updates its $high\_qc$ if $new\_view.qc.view>high\_qc.view$. The{n each} node  builds  $new\_view$ message with updated $high\_qc$ and forwards it to the  {nodes in} its parent committee.  {The latter ensures that the} latest $high\_qc$ moves upward in {the} overlay.

    \item The root committee forwards {to the leader} its $new\_view$ messages {together} with the $Aggregated\_{qc}$ from two-thirds of $new\_views$ {of} its children. 

    \item The next leader builds an aggregated $QC$ from {the} $new\_view$ messages it {has received}.  Then {it} adds this $Aggregated\_{qc}$ to the block and propose{s} th{is} block.
\end{itemize}

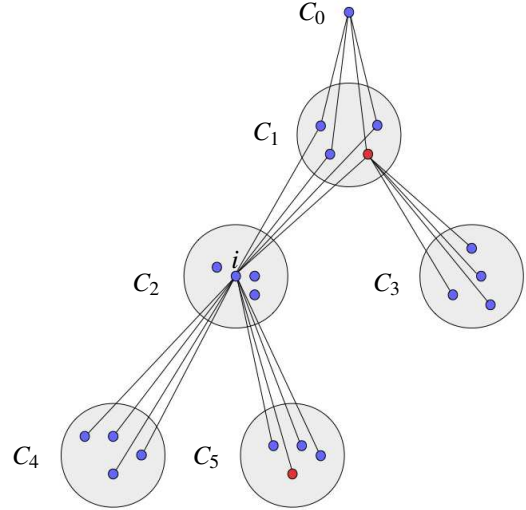
\begin{figure}[!t]
\setlength{\unitlength}{1mm}
\begin{picture}(100,75)
\put(7,0){\includegraphics[height=75\unitlength]{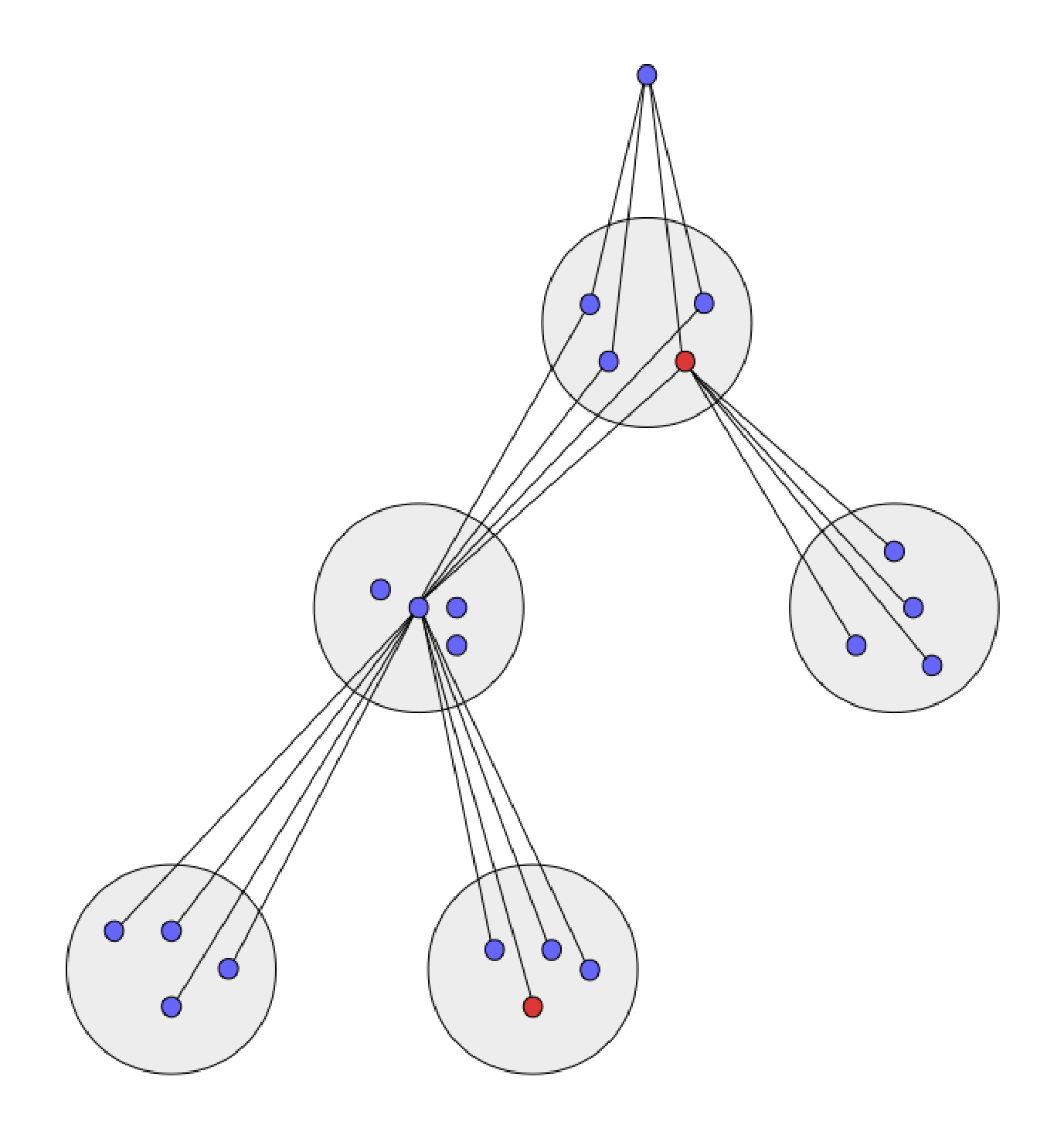}}
\put(43,69){$C_0$}
\put(37,53){$C_1$}
\put(21,33){$C_2$} 
\put(34,36){$i$}
\put(53,33){$C_3$} 
\put(5,10){$C_4$}  \put(29,10){$C_5$} 
\end{picture}
\caption{{The node $i$ in the  committee $C_2$ collects votes from its children committees, i.e. the committees   $C_4$ and $C_5$, then it casts its own vote to nodes in the parent committee  $C_1$.}}
\label{fig:comm-tree-connect} 
\end{figure}
%


\subsection{Detailed {d}escription {of the protocol}}
\label{Subsection: Detailed Description}
As a block proposer, the role of  the leader is to convince other nodes that  the proposed block  carries the latest QC  during the happy path. In this scenario if the view of the proposed block is one more than the view of the QC, i.e. $block.view == block.qc.view + 1$, then  {this} is considered  {to be the} proof that  {this} QC is the latest.

During  an unhappy path, a  leader has to make sure that  the proposed block is extending the last committed block, i.e. the latter is not averted by the proposed block.  To satisfy this condition the leader has to provide {a} proof in the form of $Aggregated\_{qc}$. The global $high\_qc$ in {the} $Aggregated\_{qc}$ is the $QC$ of the  {last} committed block in the network or a $QC$ {of a} more recent {, i.e.} with a higher view number{,} than the {last} committed block. In {the safety proof (see section \ref{Section: Proof-of-correctness})} we show how the global $high\_qc$ in {the} $Aggregated\_{qc}$ is the $QC$ of the  {last} committed block  {by} the network or a $QC$ formed after the $QC$ of the  {last} committed block. {Let us now consider the happy and unhappy paths.}

\subsubsection{{The} happy path}  {The} Algorithm \ref{alg:propose-block} is executed by  {the} leader who either receives a set of votes or a set of $new\_view$ messages. In the context of  {this} algorithm a set of votes  {indicates} that the  view {,} during which nodes received a proposal (also called  block) and have voted {,} was successful. If  the leader receives a set of votes {then} it checks if the number of votes  is greater than or equal to the $leader\_super\_majority\_threshold$ \footnote{More than two-third members of the root committee and its child committees.}. If the {latter is true then}  leader builds a QC  {, which is added to the} block along with any transactions for the current view, and proposes the block to the network by broadcasting it.



 {The} Algorithm \ref{alg:approve-block} is called when  a node receives a block from the leader and  a set of votes for the same block from {the nodes in}  its children committee. The vote set is empty if the committee  to   which the node  {belongs} has no children. The goal of  {the} algorithm is to send  a vote for {a} given block and {to} increment the voted view (to avoid duplicate voting) and view number. The block received has to be  safe,  i.e. $block.id() \in safe\_blocks$. A block is safe if {the} $block.view == block.qc.view+1$ or {the} $block.view == block.AggQC.view+1$ and $block.qc == AggQC.high\_qc$. This means {that} during {the} happy path a block is safe if $block.view == block.qc.view+1$.  {The latter} condition shows that $block.qc$ is the latest $qc$ formed before the block which is known by the majority of nodes or more recent than the recent $qc$ known by the majority of nodes. 
 
 The  {Algorithm \ref{alg:approve-block}} also performs other checks including {one which makes}  sure {that the votes  received by a node are from its children and that the }  number of votes received is at least {the} two-thirds of {the total number of its children, i.e. exceeding the }  $super\_majority\_threshold$. If checks {, i.e. the } lines $1-5$ {in the algorithm,} are successful then the node will send its vote to the parent committee. If  the node is a member of the root committee, then it will build  $QC$ from {the} votes collected from  {its} children committee{s} add {to it}  its own vote and forward it to the leader of the next view. The leader of the next view will build {a} $QC$ from {more than} {the} two third{s of} votes from {the nodes in } the root and its children committees ($leader\_super\_majority\_threshold$). Then the leader assigns  $QC$ to the $block.qc$ and proposes the block as {described} in {the} Algorithm \ref{alg:propose-block}.

 \subsubsection{{The} unhappy path}
 To ensure   low authenticator complexity, {the} Carnot differs from traditional BFT-based protocols in its approach to the unhappy path. Instead of waiting for a large number of nodes to fail, {the} Carnot detects failure when  {the} \emph{$leader\_super\_majority\_threshold$}  of nodes in the root  and its children committees fail. 

In the unhappy path, {the} nodes that are not  {in} the root {and}  its children committees stop voting after the timeout. {The nodes in}  the root   {and} its children committee{s} send timeout messages to each  {node in} the root committee {which} indicat{es} that a new view needs to be created. {Then nodes in the }  root committee build {the} timeout $QC$, {thereby} reducing the overall complexity and ensuring that the protocol can continue even in the face of node failures.

The  Algorithm {\ref{alg:timeout-detected}}  is executed by the root committee when it detects that a $leader\_super\_majority$ of nodes in the current view from  {the} root and its children committees have timed out. The purpose of  {this} algorithm is to initiate the creation of  timeout $QC$ and {to} notify all nodes in the  {network} that a timeout has occurred. The algorithm takes as input   {the} set of timeout messages received by the root committee. These timeout messages are sent by nodes{,} from the root committee and its children committees{,} when they  detect that the view (round) has failed to be executed in a timely manner.  The algorithm then performs several checks to ensure {that} the timeout messages are valid. It checks that {they} have a view number {that is} greater than or equal to the current view {and} that all  messages have the same view number. The algorithm then builds  {the} timeout QC message. This  message is  broadcasted to all nodes in the network. The purpose of  {the latter} is to notify all nodes that a new view needs to be created and to get them ready for sending their new view messages upward{, i.e. } to the parent committee members.  {The} new view messages are used  by a node to update its $high\_qc$ and {to} send the $high\_qc$ to the {nodes in}  its parent committee. In this way {the} $high\_qc$ moves to the top {of the committee tree} {and reaches} the next leader without  {the} broadcast{ing} and verification   {of}  $O(N)$ messages.
Upon receipt of {the} $timeout\_qc$ {message} a node updates its $high\_qc$, increments {the} view number, and rebuilds the overlay  {(see the} Algorithm \ref{alg:receive-timeout-qc}{)}.

{The} Algorithm \ref{alg: approve-newview} is executed by a node after  {the} Algorithm \ref{alg:receive-timeout-qc}.  {This} algorithm uses   {the} set of new view messages $new\_views$ and  {the}  timeout QC message previously received as {an} input from the set of $timeout\_qc$ {messages} used in {the} Algorithm \ref{alg:receive-timeout-qc}. 
The $new\_view$ messages  received are from  {the nodes in} a child committee{s}  {and hence} the set of $new\_view$ messages will be empty for {a node}  {in a } leaf committee.
 The  {purpose} of  Algorithm \ref{alg: approve-newview}  is to approve the new view and  make sure {that} the next leader receives {the} global $high\_qc$.  {The latter} enable{s} the leader to build a block that points to the block of the global $high\_qc$ as its parent. {Furthermore,} the algorithm  {ensures} that the view of the $new\_view.view=timeout\_qc.view+1$\new{,} where  $timeout\_qc$ is the latest $timeout\_qc$ {of the } node.
{Next}  {it} checks that the number of new view messages is greater than or equal to the supermajority threshold. {Also} the algorithm verifies that all  senders of the new view messages are  {the nodes from}  children committees 
  {and} gets the highest QC message from the new views and updates the local $high\_qc$ value. Finally, the algorithm creates  {the} new view message and broadcasts it to the nodes  {in} parent committee. If  the node is a member of the root committee {then} it sends  {this} message to the leader of the next view {, but} otherwise it sends the {new view} message to all {nodes in }  parent committee{s}. After broadcasting the new view message, the algorithm checks if a node has already incremented its voted view by {the} $local\_timeout$. If not, it increments the voted view to avoid voting (forwarding $new\_view$ messages) in this view again.

\section{The Carnot protocol: Proof of Correctness\label{Section: Proof-of-correctness}}
{Any} consensus protocol must satisfy two main properties{,} namely {the} safety and  liveness. 
{The  Carnot protocol offers a \emph{probabilistic}   guarantee of these properties as follows.}

\begin{definition}[Safety]

{If a single (correct) node in a blockchain consensus protocol commits a block at a specific position in the blockchain (blockchain height) and no other block will ever be committed at that position, even in the presence of  Byzantine nodes,  then this protocol is considered to be safe (w.h.p).} 
\end{definition}

\begin{definition}[Liveness]
A protocol is alive if {even in the presence of Byzantine nodes} it guarantees eventual progress (w.h.p). This means {that} honest nodes will always reach an agreement and the protocol will eventually terminate (w.h.p).
\end{definition}




{ The safety and liveness  of the Carnot protocol are proved, respectively, in  the Lemmas \ref{Lemma: Two-third-votes}-\ref{Lemma: Forking-Lemma} and Lemma \ref{Lemma:Liveness correct primary}.}
   {In} Lemma \ref{Lemma: Two-third-votes} we  {show} that the existence of a valid $QC$ (Quorum Certificate) for a given block $b$  {guarantees (w.h.p.) } that more than  {$2/3$} of {nodes in} the network have   vote{d} {for} $b$. This  insight from {the} Lemma \ref{Lemma: Two-third-votes} {is essential and is a cornerstone}  {of the safety proof}.   {The} Lemma {\ref{Lemma: Two-third-votes} is then used in the Lemma \ref{Lemma: Conflicting-QCs}}  to establish that the network will not produce  {$QC$} for two \emph{distinct} blocks within the same view.  {The latter} implies that  at most {one} block will  {be} approv{ed} { within a single view}. {Next the Lemmas \ref{Lemma: Two-third-votes} and \ref{Lemma: Conflicting-QCs} are used in the proof of Lemma \ref{Lemma: Forking-Lemma}.}  This lemma substantiates  that {committing} two conflicting blocks is an impossibility.  Therefore, once a block  {is committed then this block } and its lineage of ancestors  {can not be revoked}. 
{Finally, the Lemma \ref{Lemma:Liveness correct primary} guarantees that the protocol makes progress.}

\begin{lemma}
 {Suppose that the} $QC$ certificate for a proposal  was built {and the number of Byzantine nodes  is at most $M<\frac{1}{3}N$} then  at least  $\frac{2N}{3}+1$ nodes in the {committee tree $\Tree$, generated by the Algorithm \ref{alg:form-tree},} have voted for this proposal (w.h.p). 
\label{Lemma: Two-third-votes}
 \end{lemma}

 \begin{proof}
{
Let $C_{\mu}$ be a committee in  the committee tree $\Tree$ (see Definition  \ref{def:comm-tree}) and  assume that { all child committees in} $\Tree$  {are} robust (see  Definition  {\ref{def:robust-child-comm}}). We split the proof of the lemma into two parts.}
   First, we prove that if a node {$j \in C_\mu$}  has voted for {a} block $b$ in the view $v$, then at least ${\frac{2(N_{\mu_1}+N_{\mu_2})}{3}}${, where $N_\mu=\vert C_\mu\vert$,} nodes in each pair of  siblings  {$C_{\mu_1}$ and $C_{\mu_1}$} in  a  sub-tree of {the} committee  {$C_\mu$}  {have} voted for the block $b$ during the view $v$. Second{,} we prove that at least  {$\frac{2N}{3}+1$} nodes {in the network} have voted for {the} block $b$. 
   
   {It is clear that in }  a tree with only   {one} committee, i.e. $\Tree=C_1$, at least   { $\frac{2N_1}{3}+1$ } of nodes  {in the committee $C_1$} has to vote {in order} to build a $QC$  for {the} block $b$. The same is true for a tree with three committees{:}  the root committee {$C_1$} forward{s the} votes from both of its child committees, i.e. the  committees $C_2$ and $C_3$, to the next leader {(see Figure \ref{fig:comm-tree})}. For any tree $\Tree$ with more than three committees  {the first part can be proven} by induction. 

\textbf{{The} base case:} When   node $j$ receives  {the} $QC$  {of} a block proposal $b$ with the view $v$ {this implies} that {at least $\frac{2\sum_{\mu=1}^3N_\mu}{3}+1$, where $N_\mu=\vert C_\mu\vert$, } of 
nodes from {the}
root committee   and its child{ren $\cup_{\mu=1}^3C_\mu$}  have voted for $b$. {Let us define the set $C^{\voted}_\mu\subseteq C_\mu$ of nodes which voted for proposal $b$ then } the set of nodes from both child committees that have voted {is given by $C^\voted_{2} \cup C^\voted_{3}\subseteq C_{2} \cup C_{3}$. We note that} ${N^\voted_{2} + N^\voted_{3}\geq \frac{2(N_{2}+N_{3})}{3}}${, where  $N^\voted_\mu=\vert C^\voted_\mu\vert$.}
Similarly  honest node  {$i \in C^\voted_{2}$}  if it has verified the block $b$ and has received  {at least $\frac{2(N_{4}+N_{5})}{3}$} votes from its child committees  $C_{4} \cup C_{5}$. The same is true for any honest node in  {$C^\voted_{3}$}.

\textbf{{The induction step} :} At any level of the tree {$\Tree$}, if  honest node  {$j\in C_\mu$} has voted, it implies that at least {$\frac{2(N_{2\mu}+N_{2\mu+1})}{3}$}  {of} nodes from its child committees  {$C_{2\mu}$ and $C_{2\mu+1}$} have voted for {the} block $b$  {in} view $v$, or {the $C_\mu$}  is a leaf committee ({see} Algorithm \ref{alg:approve-block}). 
 {Now} by recursively applying this  {argument}  to all committees in  {the sub-tree $\Tree_\mu $ induced by $C_\mu$}  until   leaf {committees are reached, we have proved the first part. } 

{We note that since the committee tree $\Tree$ was generated by the  Algorithm \ref{alg:form-tree} this implies that the sizes of any two committees may differ by at most one node. The latter rules out  a scenario when  some of  the child committees  contain \emph{only} Byzantine nodes, but the tree $\Tree$ is robust.  
}


 {Let us now prove the second part.}  We know that a node $i \in C_{1}$ only  {sends} its vote to the leader if it {has} receive {d at least $\frac{2(N_{2}+N_{3})}{3}$} of votes from {the committees} $C_{2}$ and $C_{3}$. {Furthermore, } when  node $j$ receives  {the} $QC$ for  block $b$, then  {this implies that at least $\frac{2\sum_{\mu=1}^3N_\mu}{3}+1$, where $N_\mu=\vert C_\mu\vert$,}  of nodes in { $\cup_{\mu=1}^3C_\mu$} have voted. {Now} by combining {the number of nodes which voted from} the base case and {induction step of the first proof  we have that } the total number of nodes in $\Tree$ that have voted    {, $N^\voted$,  is } given by

\begin{equation} 
{\begin{aligned}
N^\voted & =  \vert\cup_{\mu=1}^3C^\voted_\mu\vert +\sum_{\mu=2}^{\tilde{K}} \vert C^\voted_{2\mu}\cup C^\voted_{2\mu+1} \vert \\
\end{aligned} }
\end{equation}
{from which follows}

\begin{equation} 
{\begin{aligned}
N^\voted &\geq \frac{2\sum_{\mu=1}^3N_\mu}{3}+1  + \sum_{\mu=2}^{\tilde{K}}   \frac{2(N_{2\mu}+N_{2\mu+1})}{3}\\
\end{aligned} }
\end{equation}
{and hence}
\begin{equation} 
\begin{aligned}
 { N^\voted} &\geq     \frac{2N}{3}+1  
\end{aligned} 
\label{eq:N-voted-lb}
\end{equation}

We note that for $N$ nodes with  at most $M<\frac{1}{3}N$ Byzantine nodes  {all child committees in} the tree $\Tree$  {are} robust w.h.p. (see section \ref{section:failures})   and hence above is true w.h.p. 

\end{proof}

\begin{lemma}
For any {two} valid {$QC$s, the} $qc_{1}$ {and} $qc_{2}${,} when {the}  $qc_{1}.block$ conflicts with  {the} $qc_{2}.block$ then $qc_{1}.view \not = qc_{2}.view$.
\label{Lemma: Conflicting-QCs}
\end{lemma}

\begin{proof}
{We prove this by contradiction. Assume that there are two $QC$s such that $qc_1.view == qc_2.view$ and  $N^{v_1}$ number of nodes have voted for $qc_1$. The latter implies that  $N^{v_2}=N-N^{v_1}$ number of nodes have voted for $qc_2$, but from  the Lemma \ref{Lemma: Two-third-votes} follows that $N^{v_1}\geq\frac{2}{3}N+1$ and hence $N^{v_2}\leq\frac{1}{3}N-1$. However, by the same lemma $N^{v_2}\geq\frac{2}{3}N+1$ which can be only true if at least $\frac{2}{3}N+1-(\frac{1}{3}N-1)= \frac{1}{3}N+2$ number of nodes have voted for both $QC$s. Comparing this number  with  the number of Byzantine nodes $M<\frac{1}{3}N$ we infer that  even if all Byzantine nodes have voted for both $QC$s we will have at least $\frac{1}{3}N+2 -\left(\frac{1}{3}N-1\right) =3$  honest nodes which voted for both  $QC$s. However, an honest node never votes twice in the same view, so our assumption is incorrect. Hence, this Lemma is correct for any $M < \frac{1}{3}N$}.  
\end{proof}

\begin{figure*}[!t]
\setlength{\unitlength}{1mm}
\begin{picture}(230, 55)
\put(0,0){\includegraphics[height=50\unitlength]{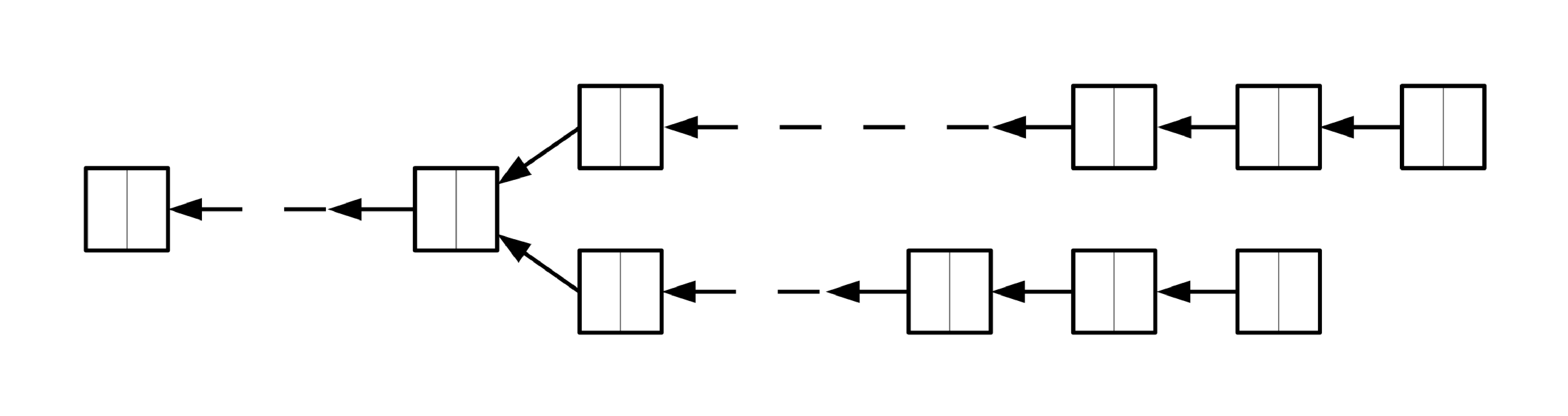}}
\put(129,34){$QC\;\;a$} \put(149,34){$QC\;a^\prime$} \put(168.5,34){$QC\;a^{\prime\prime}$}
\put(109,14){$QC\;\;b$} \put(129,14){$QC\;b^\prime$} \put(149,14){$QC\;b^{\prime\prime}$}
\end{picture}
\caption{The blocks $a$ and $b$ {in the blockchain} are conflicting blocks, but  block $b$  was committed first.}
\label{fig:fork} 
\end{figure*}
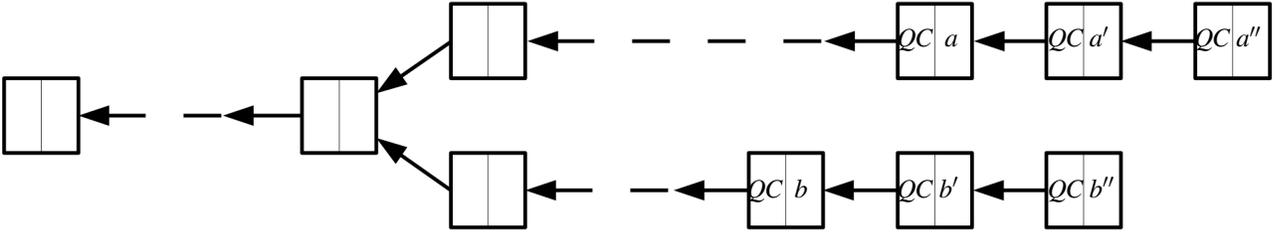

\begin{lemma}
If {the} blocks  $a$ and $b$ are in conflict, then they cannot be committed by an honest node.
\label{Lemma: Forking-Lemma}
\end{lemma}

\begin{proof}
This lemma can also be proven by contradiction. We begin by observing that block $a$ and block $b$ cannot both belong to the same view, as demonstrated in the argument of Lemma \ref{Lemma: Conflicting-QCs}. Let us assume that block $a$ is committed by an honest node during view $v$ through a two-chain (one-direct chain) sequence $a, a^{\prime}, a^{\prime\prime}$. Similarly, let's assume that block $b$ is committed through a two-chain sequence $b, b^{\prime}, b^{\prime\prime}$ (one-direct chain). Additionally, both of these two-chain sequences have their respective Quorum Certificates (QCs). Without loss of generality, we also assume that block $a$ is added to the chain at a higher view than $b^{\prime}$, meaning $a^\prime.QC.view > b^{\prime\prime}.QC.view$.

Let $v_{s}$ be the lowest view number higher than the view of $b^{\prime}$, denoted as $v_{b^{\prime}}$, such that $v_{b^{\prime}}=v_{b^{\prime\prime}}.QC.view$. During this view $v_s$, a Quorum Certificate $qc_{s}$ is formed, and $qc_{s}.block$ is in conflict with block $b$. The view number of the latter is given by $v_b =v_{b^{\prime}}.QC.view$. We will now define the following predicate:

$$E(qc) := (v_{b^{\prime}} < qc.view \leq v_a) \land (qc.block \textit{ conflicts with b})
$$
The Quorum Certificate (QC) representing the first conflicting block with respect to block $b$ is given by:
$$qc_{s} := \argmin\limits_{qc} \{  qc.view | \textit{qc is valid} \land E(qc) \}.$$

{The}  $qc_{s}$ exist {by our assumption} and {we also assume that}  $a^{\prime}.QC= qc_{s}$. Let $i$ be the node that has voted  for  {both the} $qc_{s}$ and $b^{\prime\prime}.QC$. {We note that by definition}  $qc_{s}$  is the first conflicting block {which} has the lowest view {number} among the conflicting blocks for  $b$. Therefore, block $a$ has to be in $safe\_blocks$. {Let us} now check 
{if the  conflicting block $qc_{s}.block$ is in $safe\_blocks$.}
In Algorithm \ref{alg:approve-block} the first assertion is False. This is  {because} if a block ($a$) satisfies the condition $block.view == block.qc.view+1$ then it cannot be  {in conflict with the block} $b$, {because  $block.qc$ is the $QC$ for $b$. }
{Furthermore}, let block $t$ be the parent (ancestor) of the $qc_{s}.block$.  {The definition of $qc_{s}$ implies that the} $t.view \leq b'.view$ {,} but since {the} $qc_{s}.block$ is in conflict with  {the block} $b${,} and $b$ is committed, therefore {the block} $t$   {is not the block} $b$ or $b^\prime$ {and hence} \ $t.QC.view < b^{\prime}.QC.view$. Then {by} the second  {part} of the second condition {the} $block.view == block.AggQC.view+1$ and $block.qc == AggQC.high\_qc$ is also false. This is because when a two-chain (with one direct-chain) is formed by {the blocks} $b$, $b^{\prime}$ {and} $b^{\prime\prime}$   {it implies that} at least {$2M+1$} nodes have the $QC$ ($b^{\prime}.QC.block=b$). Therefore the $high\_qc$  {,}built from $QC$s of  {$2M+1$} distinct nodes{,} collected by the new leader in the unhappy path will have $high\_qc.view \geq b.view$. 
Hence, the fact that the $i$ has voted for the $qc_{s}.block$ is incorrect {and } contradict{s} our assumption. 

\end{proof}

\begin{lemma}
\label{Lemma:Liveness correct primary}

{If a sufficient amount of time is elapsed after the GST such that a correct leader is elected and all correct nodes are in the same view then a  block $b$ will be committed  during this time.}

\end{lemma}

\begin{proof}
 {Given that a time interval after the GST is  sufficiently large,  it will be a time when two correct leaders are   elected consecutively,  { all child committees in} the overlay tree  {are} robust, and correct nodes are in the same view,  then the  one-direct chain $b,b^\prime$ will be formed between the blocks $b$ and $b^\prime$  proposed by  leaders.  Now the  two-chain $b,b^\prime,b^{\prime\prime}$ is required for the  block $b$ to get committed. However,  the  block $b^{\prime\prime}$ will eventually be added by another correct leader.} 
\end{proof}

\section{Analysis of failures\label{section:failures}}
\subsection{{Robustness of  committee trees}}
 A cornerstone of   the correctness proof
is the Lemma \ref{Lemma: Two-third-votes} which ensures that   at least $2/3$ of nodes in the  network have voted  for a proposal. The latter  relies on a  committee tree (see Figures \ref{fig:comm-tree} and \ref{fig:comm-tree-connect}) 
generated by the  Algorithm \ref{alg:form-tree}. This algorithm takes as input  the committee size  parameter,  $n$, and a set of nodes of size $N$. The committee size $n$ is computed by the Algorithm \ref{alg:comm} (see Appendix \ref{appendix:alg}) for a given number of nodes $N$, assumed fraction of Byzantine nodes, $P$, and  probability of failure  $\delta(E_1(1/3))$. The latter is the  probability that in at least one committee more than $1/3$ of nodes are Byzantine   (see Appendix \ref{appendix:struct-failure}) and it \emph{dominates} other failure probabilities  in our model (see Appendix \ref{appendix:random-part}).

\subsection{Robustness of child committees}
 The probability of  failure $\delta(E_2(1/3))$,  i.e. in at least one $C_{\mu_1} \cup C_{\mu_2}$, where the child committees $C_{\mu_1}$ and  $C_{\mu_2}$ share the same parent,  more that $1/3$ of nodes are Byzantine  (see Appendix \ref{appendix:2-child-comm-failure}), is bounded from above by $\delta(E_1(1/3))$ which is a consequence of the  Proposition \ref{Proposition:E_2-E_1-ineq} (see Appendix \ref{appendix:delta-prop}). The latter ensures that given $\delta(E_1(1/3))$  {all child committees in} a committee tree, generated by the Algorithms  \ref{alg:form-tree} and \ref{alg:comm},  {are} \emph{robust}  (see Definition  {\ref{def:robust-child-comm}} ) with the probability $\delta(E_2(1/3))\leq \delta(E_1(1/3))$.  The probability $\delta(E_2(1/3))$  can be  much smaller than   $\delta(E_1(1/3))$  as can be seen in the Figure \ref{fig:failure-prob}.  

\subsection{Logarithmic growth of committee sizes}
When $A> P$ the probability of failure $\delta(E_1(A))$, i.e. the  probability that in at least one committee more that  the  fraction  of nodes, $A$, are Byzantine   (see Appendix \ref{appendix:struct-failure}), is bounded above as follows  
\begin{equation}
\delta(E_1(A))
\leq\sum_{\mu=1}^K \rme^{-N_\mu \mathrm{D}\left(A(\mu)\vert\vert P\right)}\label{eq:delta-U-ub},
\end{equation}
where $A(\mu)=\frac{\lfloor A N_\mu\rfloor+1}{N_\mu}$, for both the hypergeometric (\ref{eq:prob-hyper}) and binomial (\ref{eq:prob-binom}) distributions of committee sizes~\cite{Mozeika2023} used in our model (see Appendix \ref{appendix:random-part}). For $N=nK+r$ nodes the  Algorithm \ref{alg:form-tree} creates the $K-r$ and $r$ committees with, respectively the sizes $n$ and $n+1$ then it is easy to show that 
$\sum_{\mu=1}^K \rme^{-N_\mu \mathrm{D}\left(A(\mu)\vert\vert P\right)}\leq K \rme^{-n\, \mathrm{D}\left(A\vert\vert P\right)}$
and hence  $\delta(E_1(A))
\leq K \rme^{-n\, \mathrm{D}\left(A\vert\vert P\right)}$ which is equivalent to  $\log(N/\delta(E_1(A))
\geq n\, \mathrm{D}\left(A\vert\vert P\right)+\log(n)$.  The latter for $n\geq n_{\min}$, where $n_{\min}$ is the \emph{minimum} number of nodes in a committee,  leads  us to  the inequality   
\begin{equation}
n\leq\frac{1}{\mathrm{D}\left(A\vert\vert P\right)}\log\left(\frac{N}{n_{\min}\,\delta(E_1(A)}\right)
\label{eq:n-ub}.
\end{equation}
Thus  the size of a committee grows at most \emph{logarithmically} with $N/\delta(E_1(A)$ (see  Figure \ref{fig:comm-size}). Also for $A=1/3$ and $P=A-\epsilon$, where $\epsilon\in(0,A)$,  the upper bound  in (\ref{eq:n-ub}) is given by  $\frac{1}{\mathrm{D}\left(1/3\vert\vert 1/3-\epsilon\right)}\log\left(\frac{N}{n_{\min}\,\delta(E_1(A)}\right)=\left[\frac{4(1-\epsilon)}{9 \epsilon^{2}}-\frac{19}{18}-\frac{4 \epsilon}{9}+O\! \left(\epsilon^{2}\right)\right]\log\left(\frac{N}{n_{\min}\,\delta(E_1(A)}\right)$ when $\epsilon\rightarrow0$ and hence for the probability of failure $\delta(E_1(A)$ to remain the same the size of a committee $n$ has to grow (at most) as $O(1/\epsilon^2)$ when the fraction of Byzantine nodes in the network approaches $1/3$. Therefore, {scalability of } the protocol is achieved through  {a} logarithmic growth of committee {sizes}  {with respect to} $N$.  As a consequence, nodes need to verify {only} $O(\log N)$ signatures, whereas each member of the root committee and the leader aggregate only $O(\log N)$ signatures.

\subsection{Necessary conditions for QC}
The Lemma \ref{Lemma: Two-third-votes} assumes that a QC certificate was built. The \emph{necessary} conditions for the latter are the election of  an honest leader,  { all child committees in}  committee tree {are robust}, and more than $2/3$ of nodes in the top three comm. have to be  honest, i.e.  we have so-called   ``leader super-majority''.  The latter is equivalent to the events  $E_0$ (election of Byzantine leader), $E_3(1/3)$ (failure of leader super-majority), and $E_2(1/3)$ (failure of  {child committee} robustness) not occurring (see Appendix \ref{appendix:failures} for a more precise definition). Now the probability $\Prob(E_0\cup E_3(1/3) \cup E_2(1/3))$, i.e. at least one of these events has occurred,  is bounded as follows 
\setlength{\arraycolsep}{0.0em}
\begin{eqnarray}
&&\Prob(E_0\cup E_2(1/3) \cup E_3(1/3))\nonumber\\
&&~~~~~~~~~~~~~~~\leq \Prob(E_0)+\Prob(E_2(1/3))+\Prob(E_3(1/3))\nonumber\\
&&~~~~~~~~~~~~~~~\leq P
+\delta(E_2(1/3))+\delta(E_3(1/3))
\label{eq:1-prob-QC-ub}
\end{eqnarray}
\setlength{\arraycolsep}{5pt}
by the union bound and hence  the probability that none of these events has occurred, i.e. necessary conditions for QC, is bounded from below as follows 
\setlength{\arraycolsep}{0.0em}
\begin{eqnarray}
&&1-\Prob(E_0\cup E_2(1/3) \cup E_3(1/3))\nonumber\\
&&~~~~~~~~~~~\geq 1-P-\delta(E_2(1/3))-\delta(E_3(1/3))\label{eq:prob-QC-lb}
\end{eqnarray}
\setlength{\arraycolsep}{5pt}
We note that $\delta(E_2(1/3))\leq \delta(E_1(1/3))$ and $\delta(E_3(1/3))\leq \delta(E_1(1/3))$ by, respectively, the Propositions \ref{Proposition:E_2-E_1-ineq} and \ref{Proposition:E_k-E_1-ineq} (see Appendix \ref{appendix:delta-prop}). Hence the lower bound in (\ref{eq:prob-QC-lb}) is greater than $1-P-2\delta(E_1(1/3))$.  Now the  probability of failure $\delta(E_1(1/3))$ can be made as small as possible  by choosing $K$ (see Figure \ref{fig:failures})  and hence the lower bound in (\ref{eq:prob-QC-lb}) is approximately $1-P$ which is greater than  $2/3$ for $P<1/3$. Thus the necessary conditions for QC are created with high probability. 

\subsection{Safety failure analysis}
The events $E_3(2/3)$ and $E_1(1/2)$, i.e., respectively,  more than $2/3$ of nodes in the top three committees and more than $1/2$ of nodes in at least one committee are Byzantine, are considered to be safety failures.  
Having  {committees with more than} $1/2$ {of nodes being} Byzantine   can jeopardize safety. {Byzantine} nodes {in such committee} might rush to approve blocks, bypassing votes from other committee members, while all nodes from  {its} sibling committee have voted.  {The latter} may {cause the subtree of }  the corrupt committee to lag behind whereas the remaining parts of the overlay tree ``move'' faster.  During {the} view change this could result in  {more than}  {$2/3$} of nodes lacking QCs for the latest committed block.  {The probability of this event, however, is very small (see Figure \ref{fig:failure-prob}).}  

The probability $\Prob(E_3(2/3) \cup E_1(1/2))$ that at least one of these events will happen is bounded from above as follows 
\setlength{\arraycolsep}{0.0em}
\begin{eqnarray}
\Prob(E_3(2/3) \cup E_1(1/2))
&\leq& \Prob(E_3(2/3))+\Prob(E_1(1/2))\nonumber\\
&&=\delta(E_3(2/3))+\delta(E_1(1/2))
\label{eq:safety-fail-prob-ub}
\end{eqnarray}
\setlength{\arraycolsep}{5pt}
by  the union bound. We note that $\delta(E_3(2/3))\leq\delta(E_3(1/3))$ and 
$\delta(E_1(1/2))\leq \delta(E_1(1/3))$ by the Proposition \ref{Proposition:E_1(A)-E_1(B)-ineq} (see Appendix \ref{appendix:delta-prop}). Also $\delta(E_3(1/3))\leq\delta(E_1(1/3))$ by the Proposition \ref{Proposition:E_k-E_1-ineq} and hence the upper bound in (\ref{eq:safety-fail-prob-ub}) is at most $2 \delta(E_1(1/3))$. However, the probabilities  $\delta(E_3(2/3))$ and $\delta(E_1(1/2))$  can be  much smaller than $2 \delta(E_1(1/3))$ as can be seen in Figure  {\ref{fig:failure-prob}}.





\begin{figure}[!t]
\setlength{\unitlength}{1mm}
\begin{center}{
\begin{picture}(100,100)
\put(0,0){\includegraphics[height=100\unitlength]{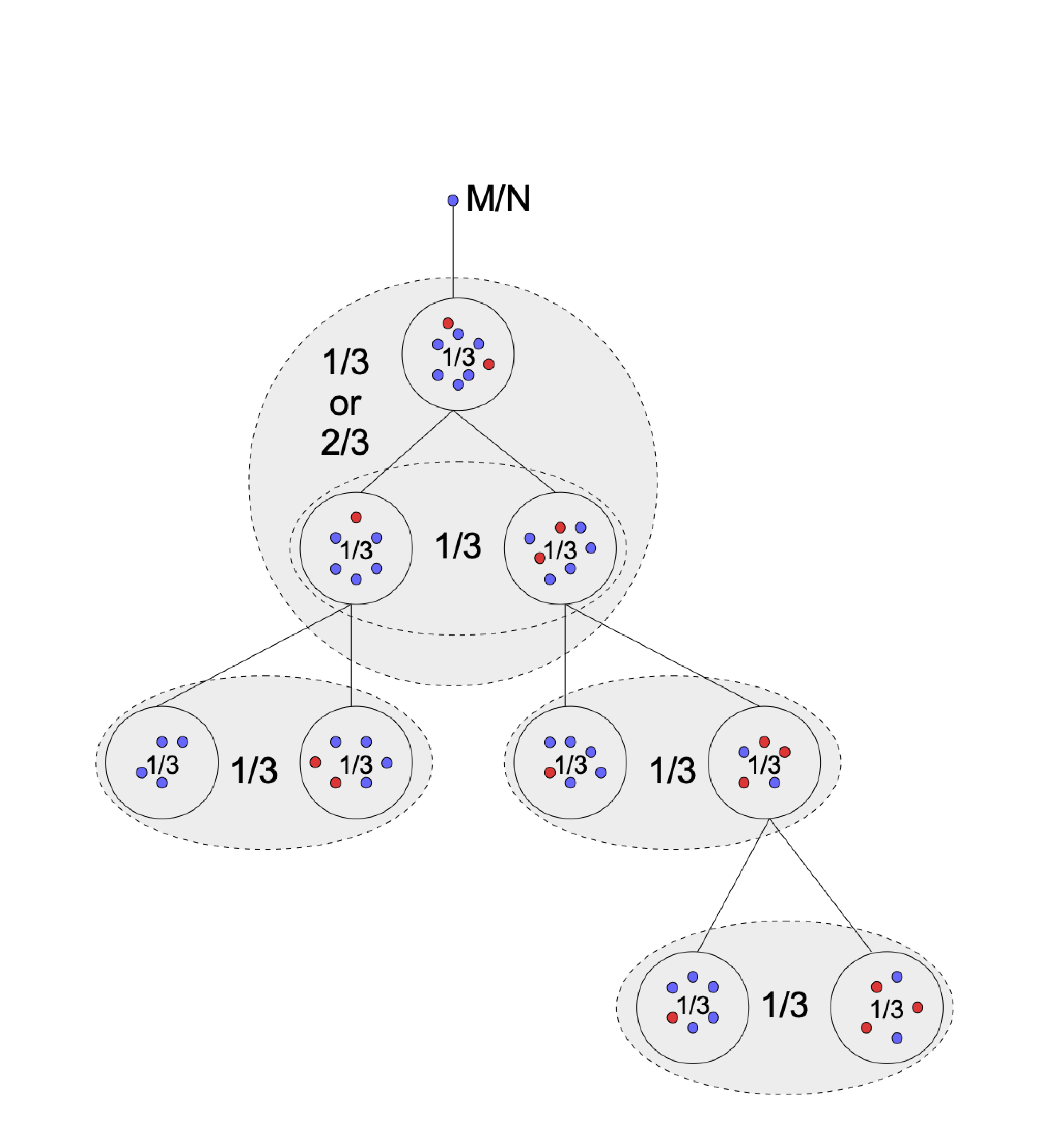}}
\end{picture}
}\end{center}
\caption{Failures in a committee tree. The $N$ nodes (blue and red dots)  are assigned \emph{randomly}  into the \emph{odd} number,  $K$,   of committees (circles with solid boundaries). Also, there is a special node representing the  \emph{leader}  (blue node connected to the root committee).  It is assumed that  \emph{at most}  $M$  nodes are adversarial  (red dots). The committee tree failure occurs when  more than $1/3$ of nodes in at least \emph{one} of  the committees are Byzantine. The child committee failure occurs when more than $1/3$ of nodes are Byzantine  in at least \emph{one} of the committees produced by merging two child committees with a common parent (ovals with dotted boundaries). The  failure in the top three committees (large circle with dotted boundary) occurs   when   more than $1/3$  (or $2/3$) of nodes in  these   committees  are  Byzantine.   The probability of electing a Byzantine leader is $M/N$. }
\label{fig:failures} 
\end{figure}
%


\section{ Conclusion}
\label{Section:Conclusion}
In conclusion, {we envisage that by combining instantaneous finality and  adaptable scalability the} Carnot {protocol will}  emerge as a trailblazing paradigm in the realm of consensus protocols.  The protocol not only achieves prompt finality but also accommodates  organic growth of networks {by} effectively mitigating  intricate challenges tied to {the} chain reorganization and fork vulnerabilities. Our  {future work} includes {further increase in } fault tolerance  and investiga{tion of }  eliminati{ng} the direct chain prerequisite for {a} block commitment. {In addition we would like to explore}  creation of  {the} multi-leader\ {and} multi-overlay Carnot variant  {which could} enhance both  performance and {resilience of the } protocol. {Finally}, {in future work we will also consider adding economic (PoS) layer to the Carnot to  ensure seamless operation within some  crypto-economic framework.} 

\section*{Acknowledgement}
{The authors would like to thank Corey Petty, Daniel Sanchez Quiros, Augustinas Bacvinka, Giacomo Pasini, \'{A}lvaro Castro-Castilla, Daniel Kaiser and  Frederico Teixeira for their invaluable assistance and very enlightening discussions that greatly enriched this work.}

\appendices
\section{Details of analysis of failures\label{appendix:failures}}

\subsection{Model of random partitions\label{appendix:random-part}}
We consider $N$ nodes  distributed into $K$  committees in a random and unbiased way (see Figure \ref{fig:comm-tree}) for a given  committee sizes  $\{N_1,\ldots,N_K\}$, where  $N_\mu=\vert C_\mu\vert$ is the number of nodes in a committee  $C_\mu$.  Assuming that the sampling without replacement is used  and  that $M<N$ nodes  are  \emph{Byzantine}  gives us the   hypergeometric  distribution 
\begin{equation}
\Prob\left(N_{1}^\alpha,\ldots, N_{K}^\alpha\vert N_{1},\ldots,N_{K}; M\right)=\frac{\delta_{M;\sum_{\mu=1}^K N_{\mu}^\alpha}\prod_{\mu=1}^K {N_\mu\choose N^\alpha_\mu}}{{N\choose M}}\label{eq:prob-hyper}
\end{equation}
for  the $\{N_{1}^\alpha,\ldots, N_{K}^\alpha\}$, where $N_{\mu}^\alpha$ is the  number of Byzantine nodes  in the  committee $\mu$. Furthermore, if  $M$ is  \emph{random} {variable} from  the binomial distribution  
\begin{equation}
\Prob\left(M\vert N \right)={N\choose M}\,P^{M}  \left(1-P\right)^{N-M}\label{def:binom},
\end{equation}
where $P\in(0,1)$, then the \emph{average} of (\ref{eq:prob-hyper}) over $M$ is  the product  of binomials
\begin{equation}
\Prob\left(N_{1}^\alpha,\ldots, N_{K}^\alpha\vert N_{1},\ldots,N_{K}\right)=
\prod_{\mu=1}^K\Prob\left(N^\alpha_\mu\vert N_\mu \right)\label{eq:prob-binom}.
\end{equation}
We note that in \emph{sharding} of blockchains the number of Byzantine nodes in a \emph{single} committee  is usually  modeled with the hypergeometric~\cite{Zamani2018, Dang2019, Zhang2022} and binomial~\cite{Kokoris2018,Tennakoon2022} probability distributions which are, respectively, the marginals of the probability distributions (\ref{eq:prob-hyper}) and (\ref{eq:prob-binom}).  The latter two are special cases of a  more \emph{general} distribution $\sum_{M\geq 0}\Prob\left(N_{1}^\alpha,\ldots, N_{K}^\alpha\vert N_{1},\ldots,N_{K}; M\right)\Prob(M)$. 

\subsection{Analysis of committee tree failure\label{appendix:struct-failure}}
We are interested in the probability  of various  ``failures'' that can occur in a committee tree  (see Figure \ref{fig:failures}). The first failure we consider is  the event $E_1(A)=\cup_{\mu=1}^K\left\{N_{\mu}^\alpha\geq \lfloor AN_{\mu}\rfloor+1\right\}$, i.e. more than fraction $A$ of nodes in at least \emph{one} of  the committees are  \emph{Byzantine}. This type of failure concerns  \emph{reliability} of the  whole \emph{structure} of a committee tree and its probability is given by  
\begin{equation}
\delta(E_1(A))=\Prob\left(\cup_{\mu=1}^K\left\{N_{\mu}^\alpha\geq \lfloor AN_{\mu}\rfloor+1\right\}\right)\label{def:delta-E1},
\end{equation}
via  $\Prob(E_1(A))=1-\Prob(\neg E_1(A))$, is exactly equal to     
\begin{equation}
\delta(E_1(A)) =1-\Prob\left(N_{1}^\alpha\leq\lfloor A N_{1}\rfloor ,\ldots, N_{K}^\alpha\leq\lfloor A N_{K}\rfloor\right)\label{eq:delta-E1}.
\end{equation}
We note that  above was studied in~\cite{Mozeika2023} for both  the hypergeometric (\ref{eq:prob-hyper}) and binomial (\ref{eq:prob-binom}) distributions. For the latter we obtain  
\begin{equation}
\delta(E_1(A))=1-\prod_{\mu=1}^K\left[1-\Prob(N^\alpha\geq \lfloor A N_\mu\rfloor+1\vert N_\mu)\right]\label{eq:delta-binom},
\end{equation}
where 
\begin{equation}
\Prob(N^\alpha\geq \lfloor A N_\mu\rfloor+1)=
 \sum_{N^\alpha=\lfloor AN_{\mu}\rfloor+1}^{N_\mu}\Prob\left(N^\alpha\vert N_{\mu}\right)
\label{eq:binom-tail}
\end{equation}
is the upper tail of the binomial distribution (\ref{def:binom}), and by the \emph{Theorem} 3.3 in~\cite{Mozeika2023} for $P<A(\mu)<1$, where $A(\mu)=\frac{\lfloor A N_\mu\rfloor+1}{N_\mu}$, we have the upper bound   
\setlength{\arraycolsep}{0.0em}
\begin{eqnarray}
&&\delta(E_1(A))\nonumber\\
&&~~~\leq 1-\prod_{\mu=1}^K\left[1-\frac{1}{1-r(\mu)}\frac{\rme^{-N_\mu \mathrm{D}\left(A(\mu)\vert\vert P\right)}}{)\sqrt{2\pi A(\mu)\left(1-A(\mu)\right)N_\mu}}\right]\label{eq:delta-binom-ub},
\end{eqnarray}
\setlength{\arraycolsep}{5pt}
where $\mathrm{D}\left(A\vert\vert P\right)$ with  $A,P\in(0,1)$ is the Kullback$-$Leibler  (KL) `distance'
\begin{equation}
\mathrm{D}\left(A\vert\vert P\right)=A\log\frac{A}{P}+(1-A)\log\frac{1-A}{1-P}
\end{equation}
which is $0$  when $A=P$ and is positive semi-definite when $A\neq P$~\cite{Cover2012}. Here the term $\frac{1}{1-r(\mu)}$, where $r(\mu)=\frac{P\left(1-A(\mu)\right)}{ A(\mu)\left(1-P\right)}$, is an  upper bound on the sum  $\sum_{k=0}^{\left(1-A(\mu)\right)N_\mu}r^k(\mu)$.

To estimate  the probability of failure (\ref{def:delta-E1}) for the hypergeometric distribution (\ref{eq:prob-hyper}) we will use the Boole's inequality
\begin{equation}
\Prob\left(\cup_{\mu=1}^K\left\{N_{\mu}^\alpha\geq \lfloor AN_{\mu}\rfloor+1\right\}\right)
\leq\sum_{\mu=1}^K\Prob\left(N_{\mu}^\alpha\geq \lfloor AN_{\mu}\rfloor+1\right)\label{eq:Boole-ineq}, 
\end{equation}
also known as  the \emph{union bound},  which gives us 
\begin{equation}
\delta(E_1(A)) \leq\sum_{\mu=1}^K\Prob\left(N_{\mu}^\alpha\geq \lfloor AN_{\mu}\rfloor+1\right)\label{eq:delta-hyper-ub}, 
\end{equation}
where 
\begin{equation}
\Prob\left(N_{\mu}^\alpha\geq \lfloor AN_{\mu}\rfloor+1\right)=\sum_{N^\alpha_\mu=\lfloor AN_{\mu}\rfloor+1}^{N_\mu}\frac{{N_\mu\choose N^\alpha_\mu}{N-N_\mu\choose M-N^\alpha_\mu}}{{N\choose M}}
\label{eq:hyper-tail}
\end{equation}
is the upper tail of the (univariate) hypergeometric distribution, i.e.  the marginal of  (\ref{eq:prob-hyper}). 

\begin{figure}[!t]
\setlength{\unitlength}{0.57mm}
\begin{center}{
\begin{picture}(100,100)
\put(0,0){\includegraphics[height=100\unitlength,width=100\unitlength]{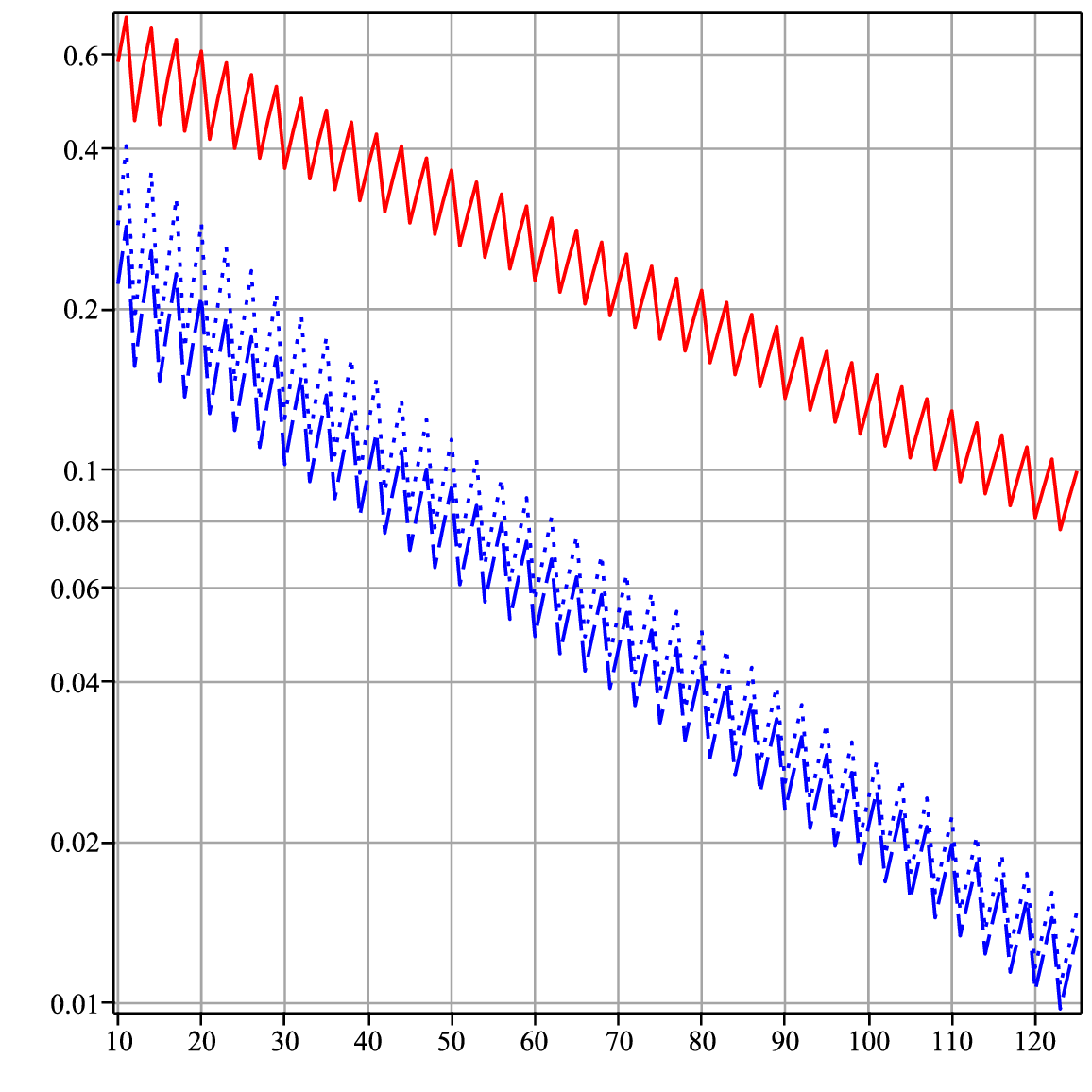}}%
\put(-10,31){\rotatebox{90}{$\Prob\left(N_{\mu}^\alpha\geq \lfloor AN_{\mu}\rfloor+1\right)$}} 
 \put(53,-5){$N_\mu$}
\end{picture}
}\end{center}
\vspace*{2mm} 
\caption{The tail of hypergeometric distribution $\Prob\left(N_{\mu}^\alpha\geq \lfloor AN_{\mu}\rfloor+1\right)$  as function of $N_{\mu}$ computed for $N=10^3$,  $M=N/4$ and  $A=1/3$. The upper bound (\ref{eq:hyper-tail-ub}) and exact numerical values  (\ref{eq:hyper-tail}) are represented, respectively, by the  blue dotted and dashed lines. The red solid line is the Hoeffding bound $\rme^{-N_\mu \mathrm{D}\left(A(\mu)\vert\vert M/N\right)}$, where $A(\mu)=\frac{\lfloor A N_\mu\rfloor+1}{N_\mu}$.} 
\label{fig:hyper-tail-ub}
\end{figure}

The numerical computation of (\ref{eq:delta-hyper-ub}) can be very inefficient because of combinatorial objects in (\ref{eq:hyper-tail}). The latter can be estimated by the Hoeffding bound~\cite{Hoeffding1963, Chvatal1979} but this bound is very loose as can be seen in Figure \ref{fig:hyper-tail-ub}.  A much tighter bound (see Figure \ref{fig:hyper-tail-ub}) is given by the 
\begin{lemma}
\label{Lemma:hyper-tail-ub} 
Suppose that $P+1/N<A$, where $P=M/N$, then 
\setlength{\arraycolsep}{0.0em}
\begin{eqnarray}
&&\Prob\left(N_{\mu}^\alpha\geq \lfloor AN_{\mu}\rfloor+1\right)\nonumber\\
&&~~~~~~~\leq \frac{{N-N_{\mu}\choose M-\left\lfloor A N_{\mu}\right\rfloor-1}{N_{\mu}\choose \left\lfloor A N_{\mu}\right\rfloor+1}}{{N\choose M}}\frac{1-r^{N_{\mu}- \left\lfloor A N_{\mu}\right\rfloor}}{1-r}
\label{eq:hyper-tail-ub},
\end{eqnarray}
\setlength{\arraycolsep}{5pt}
where $r=\frac{\left[P-\frac{\left\lfloor A N_{\mu}\right\rfloor+1}{N}\right][1-A(\mu)]      }{\left[1-P-\frac{N_{\mu}}{N}+\frac{\left\lfloor A N_{\mu}\right\rfloor+1}{N}\right]\,A(\mu)}$ and $A(\mu)=\frac{\lfloor A N_\mu\rfloor+1}{N_\mu}$. 
\end{lemma}
\begin{proof}
See Appendix \ref{appendix:proof}. 
\end{proof}
We note that the RHS in  (\ref{eq:hyper-tail-ub}) can be further bounded  by using~\cite{Macwilliams1977}    %
\setlength{\arraycolsep}{0.0em}
\begin{eqnarray}
&&\sqrt{\frac{1}{8P\left(1-P\right)N}}\rme^{N H\left(P\right)}\leq{N\choose M}\nonumber\\
&&~~~~~~~~~~~~~~~~~~\leq \sqrt{\frac{1}{2\pi P\left(1-P\right)N}}
\rme^{N H\left(P\right)},
\end{eqnarray}
\setlength{\arraycolsep}{5pt}
where $H(P)=-P\log(P)-(1-P)\log(1-P) $ is the  \emph{entropy}, in  the binomial coefficients which will reduce  numerical complexity of estimating (\ref{eq:hyper-tail}).

\begin{figure*}[!t]
\setlength{\unitlength}{0.57mm}
\begin{center}
\begin{picture}(200,200)
\put(0,105){\includegraphics[height=100\unitlength,width=100\unitlength]{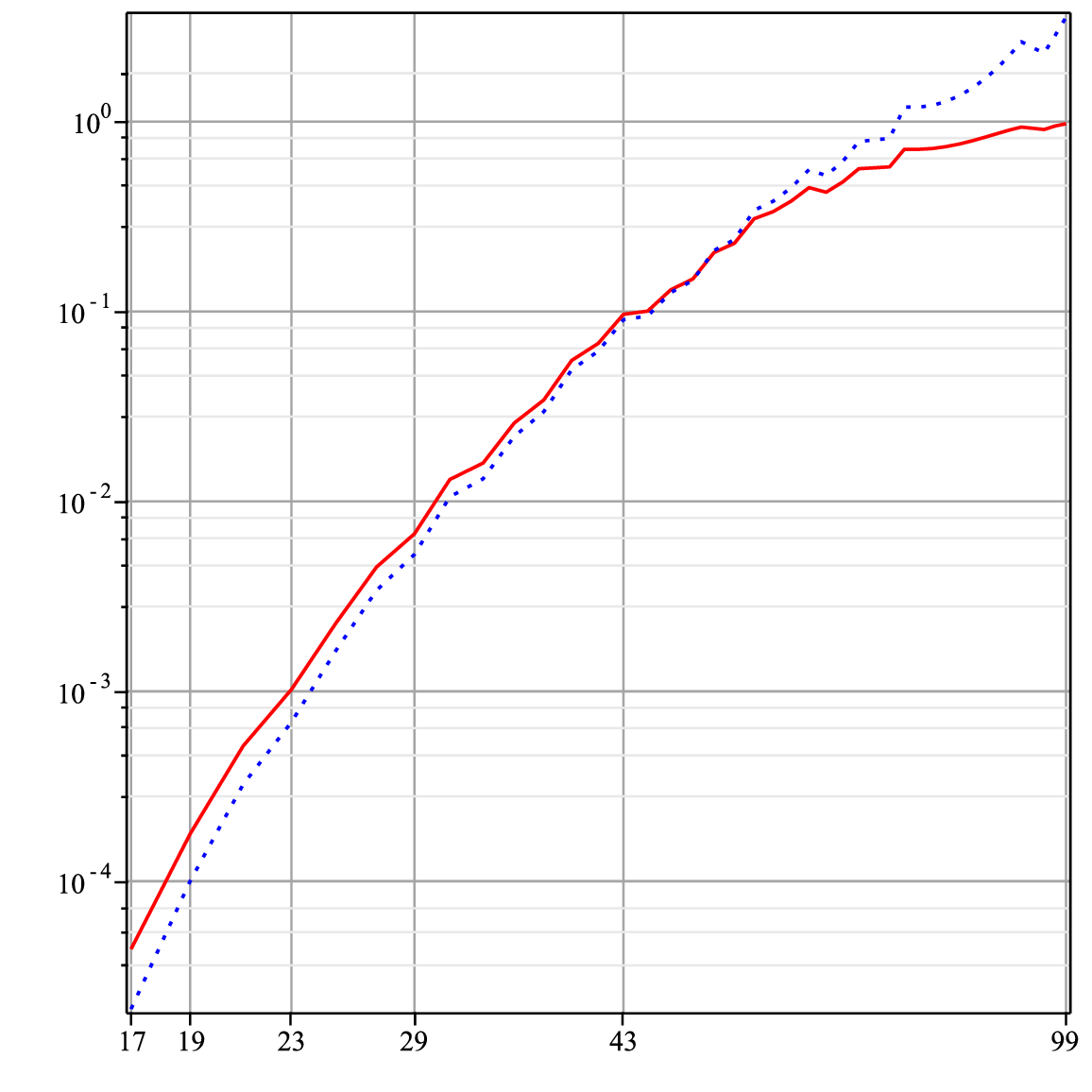}}%
\put(105,105){\includegraphics[height=100\unitlength,width=100\unitlength]{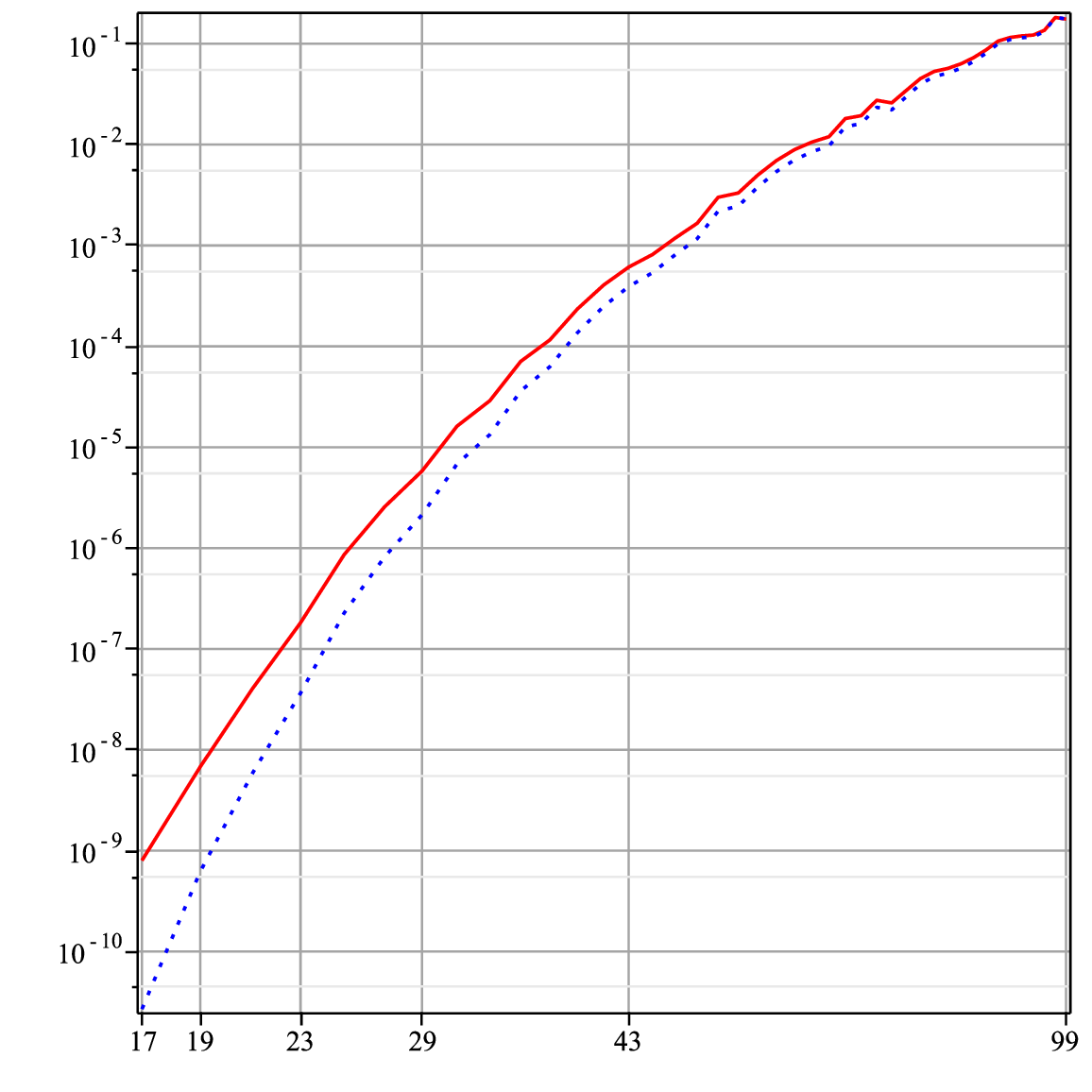}}
\put(0,-0){\includegraphics[height=100\unitlength,width=100\unitlength]{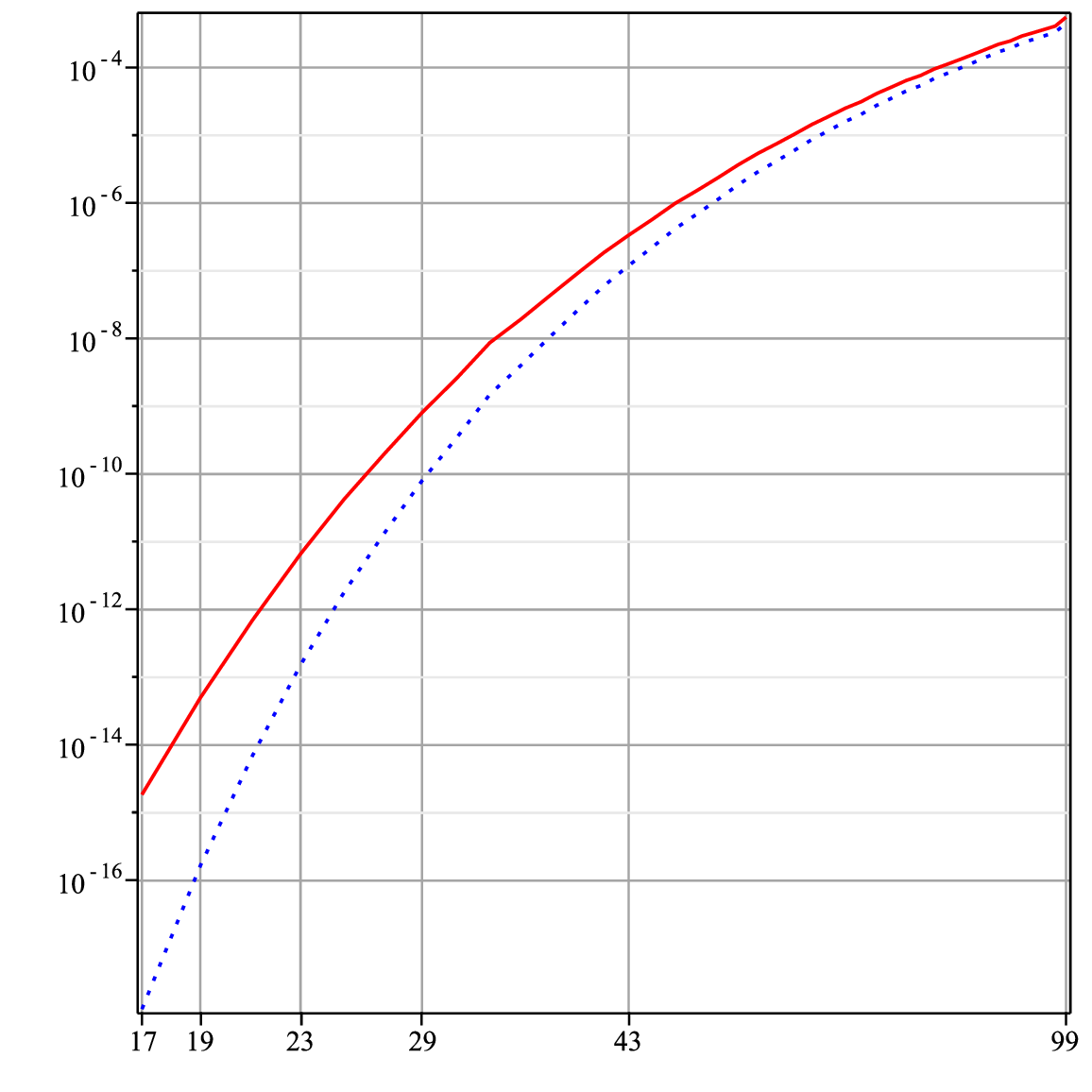}}
\put(105,0){\includegraphics[height=100\unitlength,width=100\unitlength]{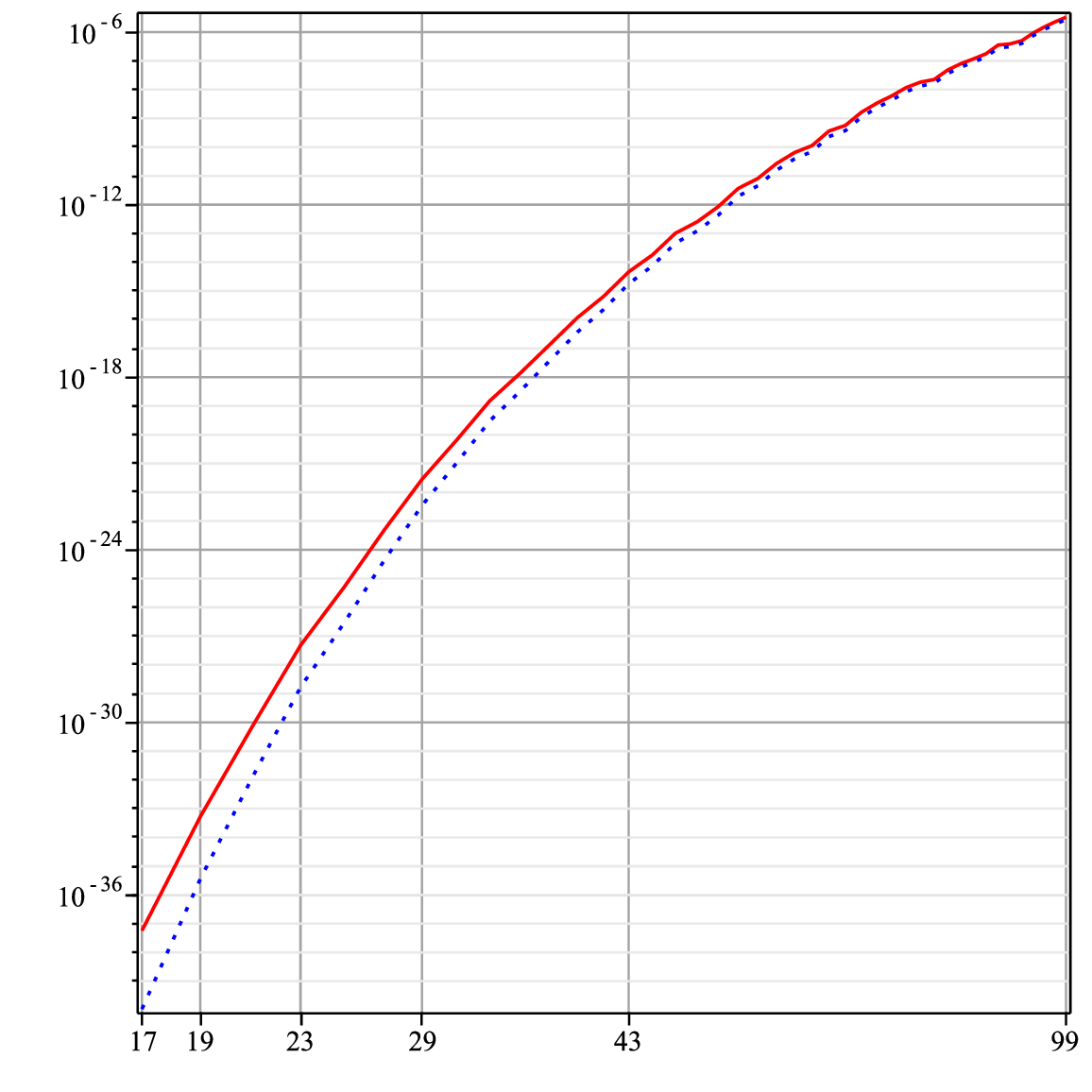}}

\put(-2,55){\rotatebox{90}{$\delta(E_3(1/3))$}} 
\put(103,55){\rotatebox{90}{$\delta(E_1(1/2))$}}
\put(-2,155){\rotatebox{90}{$\delta(E_1(1/3))$}}
\put(103,155){\rotatebox{90}{$\delta(E_2(1/3))$}}
 \put(155,-5){$K$}
 \put(50,-5){$K$}
\end{picture}
\end{center}
\vspace*{4.0mm} 
\caption{The probability of failure, $\delta$,  as a function of the number of committees, $K$, computed for the  binomial (red solid line)  and hypergeometric (blue dotted line) models of random partitions.  The number of nodes is  $N=10^4$, where $N=nK+r$, and the number of Byzantine  nodes is  \emph{exactly} (on \emph{average})  $M=N/4$ in  hypergeometric (binomial) model. 
The  $n$ and $n+1$ nodes are  assigned, respectively, to the  $K-r$ and $r$ committees.  Top left:  The probability of  failure $\delta(E_1(1/3))$ estimated by the upper bounds (\ref{eq:delta-binom-ub}) and (\ref{eq:delta-hyper-ub}). The latter uses (\ref{eq:hyper-tail-ub}). Top right: The probability of  failure $\delta(E_2(1/3))$ estimated by the upper bounds (\ref{eq:delta-E2-binom-ub}) and (\ref{eq:delta-E2-hyper-ub}). Bottom left: The probability of   failure $\delta(E_3(1/3))$ estimated by the upper bounds (\ref{eq:delta-Ek-binom-ub}) and (\ref{eq:delta-Ek-hyper-ub}). On the same range 
the probability of   failure $\delta(E_3(2/3))$ is at most $7.31\times 10^{-53}$. Bottom right: The probability of failure $\delta(E_1(1/2))$ estimated by the upper bounds (\ref{eq:delta-binom-ub}) and (\ref{eq:delta-hyper-ub}). The latter uses (\ref{eq:hyper-tail-ub}).
} 
\label{fig:failure-prob}
\end{figure*}
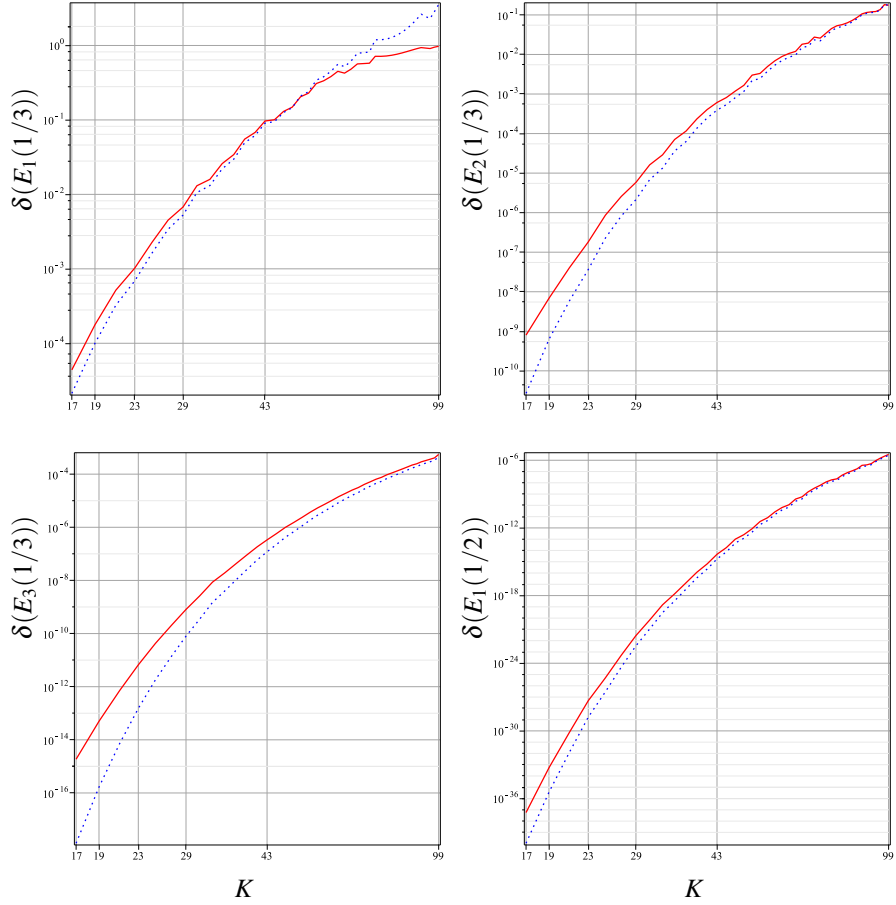
The union bound  (\ref{eq:delta-hyper-ub}) suggest that for $K<\infty$ the probability of failure $\delta(E_1(A))\rightarrow0$  if we assume that $N_\mu/N>0$ for all $\mu$ as $N\rightarrow\infty$. This follows from $\delta(E_1(A))\leq \sum_{\mu=1}^K\rme^{-N_\mu \mathrm{D}\left(A(\mu)\vert\vert M/N\right)}$ which is obtained by applying  the Hoeffding bound~\cite{Hoeffding1963, Chvatal1979} to the tail  (\ref{eq:hyper-tail}). However, outside of this asymptotic regime one can always find the number of committees $K$ and committee sizes $N_\mu$ for such that the RHS of  (\ref{eq:delta-hyper-ub}) exceeds \emph{unity}. The latter can be seen in the Figure \ref{fig:failure-prob} where we  plot the bounds (\ref{eq:delta-binom-ub}) and (\ref{eq:delta-hyper-ub}) together.  The upper bound (\ref{eq:delta-binom-ub}), which assumes the binomial distribution (\ref{eq:prob-binom}), is bounded from above by unity.  Also we expect the latter, which assumes binomial  distribution (\ref{eq:prob-binom}), to be an upper bound on  the $\delta(E_1(A))$ when  the hypergeometric distribution (\ref{eq:prob-hyper}) is assumed~\cite{Mozeika2023}.   
\begin{figure}[!t]
\setlength{\unitlength}{0.57mm}
\begin{center}{
\begin{picture}(100,100)
\put(0,0){\includegraphics[height=100\unitlength,width=100\unitlength]{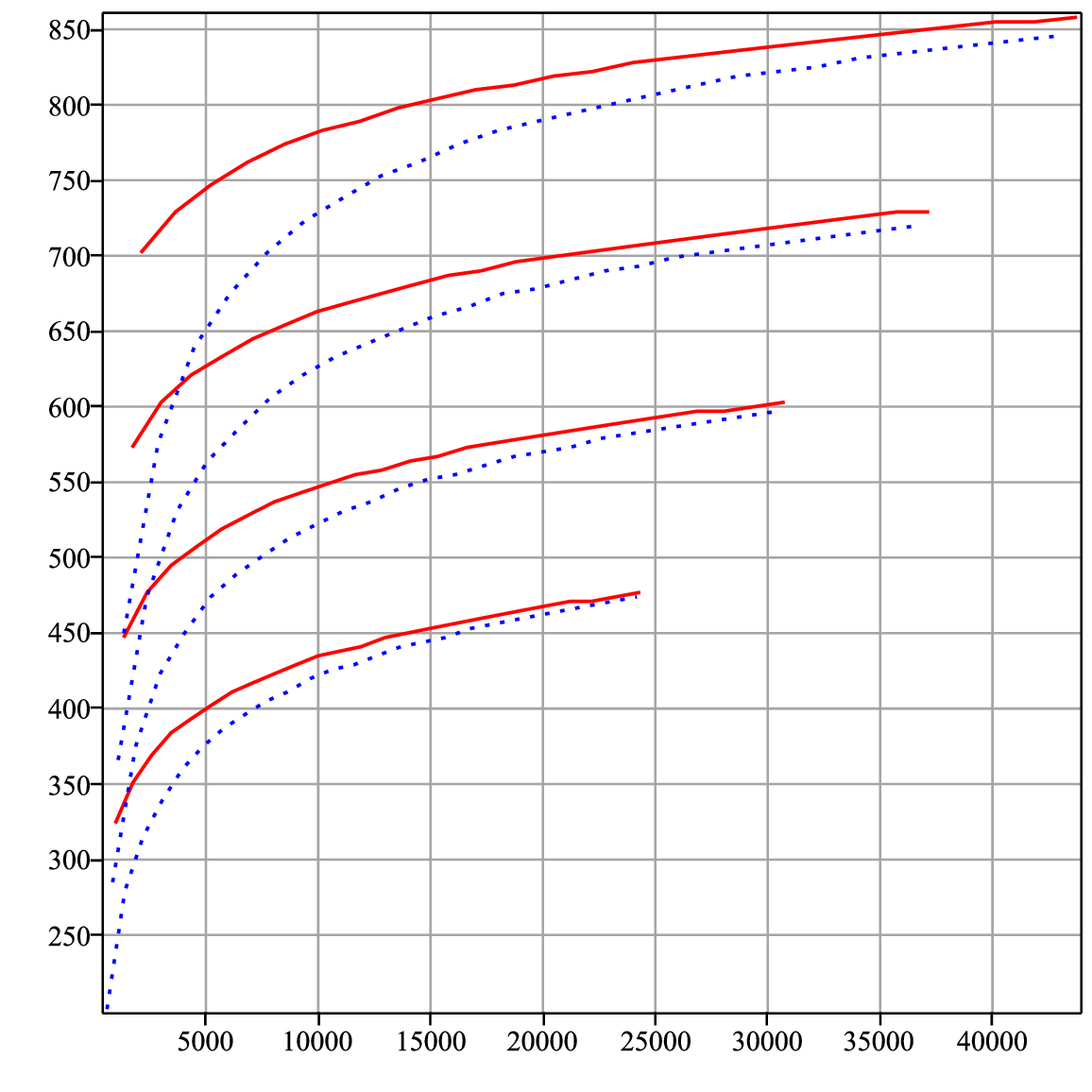}}%
\put(-2,65){$n$} 
 \put(50,-5){$N$}
\end{picture}
}\end{center}
%
\caption{The committee size, $n$,  as a function of the number of nodes,  $N$,  plotted for the probability of failure $\delta(E_1(1/3))\in\{10^{-3},10^{-4},10^{-5},10^{-6}\}$ (bottom to top) and the fraction of Byzantine nodes $P=1/4$.  The red solid and blue dotted lines were computed from, respectively, the binomial (\ref{eq:delta-binom-ub})  and hypergeometric (\ref{eq:delta-hyper-ub}) upper bounds. The latter uses the bound (\ref{eq:hyper-tail-ub}).} 
\label{fig:comm-size}
\end{figure}

 For the number of nodes $N=nK$ we can use the upper bounds (\ref{eq:delta-binom-ub}) and (\ref{eq:delta-hyper-ub}) to find the committee size $n$ for the  given  failure probability $\delta(E_1(A))$ and fraction  of Byzantine nodes $P$.  We find that  the committee size  $n$  is growing  slowly with the number of nodes $N$ and  failure probability $\delta(E_1(A))$   (see Figure \ref{fig:comm-size}).  We note that the $n$ computed with (\ref{eq:delta-binom-ub}) is an upper bound for the $n$ computed with (\ref{eq:delta-hyper-ub}), which reproduces results of the previous study~\cite{Mozeika2023}, but the difference between  the two  is vanishing  as $N$ increases (see Figure \ref{fig:comm-size}). 

Finally,  the upper bounds (\ref{eq:delta-binom-ub}) and (\ref{eq:delta-hyper-ub}) can be used to find the  \emph{maximum} number of committees $K$ for a given $N$,   $\delta(E_1(A))$ and $P$. The details of an algorithm that implements the latter are  provided  in the Appendix \ref{appendix:alg}.

\subsection{Analysis of  {child committees} failure \label{appendix:2-child-comm-failure}}
  The next type of failure we consider is when  in at least one of the  committees produced by merging two child committees with a common parent (see Figure \ref{fig:failures}) the fraction of Byzantine nodes is greater than $A$.  The latter is the event  $E_2(A)=\cup_{\mu=2}^{\tilde{K}}\left\{S_{\mu}^\alpha\geq \lfloor A S_{\mu}\rfloor+1\right\}$, where  $S_{1}^\alpha=N_{1}^\alpha$, $S_{2}^\alpha=N_{2}^\alpha+N_{3}^\alpha$, $S_{3}^\alpha=N_{4}^\alpha+N_{5}^\alpha$ etc. with $\tilde{K}=\frac{K+1}{2}$, and hence  to compute  the  failure probability
\begin{equation}
\delta(E_2(A))=\Prob\left(\cup_{\mu=2}^{\tilde{K}}\left\{S_{\mu}^\alpha\geq \lfloor A S_{\mu}\rfloor+1\right\}\right)\label{def:delta-E2}
\end{equation}
we need to know the distribution of $\{S_{2}^\alpha,\ldots,S_{\tilde{K}}^\alpha\}$. The latter for the hypergeometric distribution (\ref{eq:prob-hyper}) is given   %
\begin{equation}
\Prob\left(S_{1}^\alpha,\ldots, S_{\tilde{K}}^\alpha\vert S_{1},\ldots,S_{\tilde{K}}; M\right)=\frac{\delta_{M;\sum_{\mu=1}^{\tilde{K}} S_{\mu}^\alpha}\prod_{\mu=1}^{\tilde{K}} {S_\mu\choose S^\alpha_\mu}}{{N\choose M}}\label{eq:prob-hyper-2-child}
\end{equation} 
and  for  the  binomial distribution (\ref{eq:prob-binom}) is given by 
\begin{equation}
\Prob\left(S_{1}^\alpha,\ldots, S_{\tilde{K}}^\alpha\vert S_{1},\ldots,S_{\tilde{K}}\right)=\prod_{\mu=1}^{\tilde{K}}\Prob\left(S^\alpha_\mu\vert S_\mu \right)\label{eq:prob-binom-2-child}.
\end{equation} 

The probability of failure (\ref{def:delta-E2}) is also equal to 
\begin{equation}
\delta(E_2(A))=1-\Prob\left(S_{2}^\alpha\leq\lfloor A S_{2}\rfloor ,\ldots, S_{\tilde{K}}^\alpha\leq\lfloor A S_{\tilde{K}}\rfloor\right)\label{eq:delta-E2}
\end{equation}
which gives us 
\begin{equation}
\delta(E_2(A))=1-\prod_{\mu=2}^{\tilde{K}}\sum_{S^\alpha_\mu=0}^{\lfloor A S_{\mu}\rfloor}\Prob\left(S^\alpha_\mu\vert S_\mu \right)\label{eq:delta-E2-binom-2}
\end{equation}
for the  binomial distribution (\ref{eq:prob-binom-2-child}) and hence by the \emph{Theorem} 3.3 in~\cite{Mozeika2023} for $P<A(\mu)<1$, where $A(\mu)=\frac{\lfloor A S_\mu\rfloor+1}{S_\mu}$, we have the upper bound   
\setlength{\arraycolsep}{0.0em}
\begin{eqnarray}
&&\delta(E_2(A))\nonumber\\
&&~~\leq 1-\prod_{\mu=2}^{\tilde{K}}\left[1-\frac{1}{1-r(\mu)}\frac{\rme^{-S_\mu \mathrm{D}\left(A(\mu)\vert\vert P\right)}}{)\sqrt{2\pi A(\mu)\left(1-A(\mu)\right)S_\mu}}\right]\label{eq:delta-E2-binom-ub},
\end{eqnarray}
\setlength{\arraycolsep}{5pt}
where $r(\mu)=\frac{P\left(1-A(\mu)\right)}{ A(\mu)\left(1-P\right)}$. For the hypergeometric distribution (\ref{eq:prob-hyper-2-child}) we can use the Boole's inequality 
\begin{equation}
\Prob\left(\cup_{\mu=2}^{\tilde{K}}\left\{S_{\mu}^\alpha\geq \lfloor A S_{\mu}\rfloor+1\right\}\right)\leq 
\sum_{\mu=2}^{\tilde{K}}\Prob\left(S_{\mu}^\alpha\geq \lfloor A S_{\mu}\rfloor+1\right)\label{eq:delta-E2-U-b}
\end{equation}
to obtain the upper bound  
\begin{equation}
\delta(E_2(A))\leq\sum_{\mu=2}^{\tilde{K}}\sum_{S_\mu^\alpha=\left\lfloor A S_{\mu}\right\rfloor+1}^{S_\mu}\frac{{N-S_{\mu}\choose M-S_{\mu}^\alpha}  {S_\mu\choose S_{\mu}^\alpha} }{{N\choose M}}\label{eq:delta-E2-hyper-ub}. 
\end{equation}
We note that above bound can be simplified by using the inequality (\ref{eq:hyper-tail-ub}) to bound tails of hypergeometric distributions.  The bounds (\ref{eq:delta-E2-binom-ub}) and (\ref{eq:delta-E2-hyper-ub}) are plotted in the Figure \ref{fig:failure-prob}.

\subsection{Analysis  of failures in top three committees}  
Finally, the last type of failure we consider is when the fraction of Byzantine nodes in the top three committees is greater than $B$ (see Figure \ref{fig:failures}). The latter is a   special case of the event  $E_k(B)=\Ind\left[\sum_{\nu=1}^kN_{\nu}^\alpha\geq \left\lfloor B \sum_{\nu=1}^kN_{\nu}\right\rfloor+1\right]$, where $k\geq3$, that  the total number of Byzantine  nodes in the \emph{top}  $k$ committees  exceeds the  fraction $B$ of the total number of nodes in these committees. Here,  we assumed,  without loss of generality,  that nodes of the top $k$ committees are labeled by $[k]$.

Let us define the probability of failure 
\begin{equation}
\delta(E_k(B)) =\Prob
\left(\sum_{\nu=1}^kN_{\nu}^\alpha\geq \left\lfloor B \sum_{\nu=1}^kN_{\nu}\right\rfloor+1\right)\label{def:delta-Ek}
\end{equation}
 and consider above for the hypergeometric   (\ref{eq:prob-hyper}) and binomial (\ref{eq:prob-binom}) probability distributions. For the latter we obtain the probability  
\begin{equation}
\delta(E_k(B))
=\sum_{S_k^\alpha=\left\lfloor B S_{k}\right\rfloor+1}^{S_k}{S_k\choose S_k^\alpha}\,P^{S^\alpha_\mu}  \left[1-P\right]^{S_k-S_k^\alpha}
\label{eq:delta-Ek-binom},
\end{equation} 
where we have defined the sums  $S_k^\alpha=\sum_{\nu=1}^kN_{\nu}^\alpha$ and $S_k=\sum_{\nu=1}^kN_{\nu}$.  We note that above results  is a direct consequence of that the probability  (\ref{eq:prob-binom}) describes \emph{independent} (binomial) random variables (r.v.) and hence the  sum $S_k^\alpha$  is also  r.v. from  a   binomial distribution.  Furthermore, for the probability (\ref{eq:delta-Ek-binom}) we have the following bound   
\begin{equation}
\delta(E_k(B))\leq \frac{1}{1-r}\frac{\rme^{-S_k \mathrm{D}\left(B(k)\vert\vert P\right)}}{\sqrt{2\pi B(k)\left(1-B(k)\right)S_k}}\label{eq:delta-Ek-binom-ub},
\end{equation} 
where $B(k)=\frac{\left\lfloor B S_{k}\right\rfloor+1}{S_{k}}$ and $r=\frac{P\left(1-B(k)\right)}{ B(k)\left(1-P\right)}$, by the \emph{Theorem} 1 in~\cite{Ferrante2021}.

For the  hypergeometric  probability  distribution (\ref{eq:prob-hyper})  with  a bit of work  we obtain 
\begin{equation}
\delta(E_k(B))
= \sum_{S_k^\alpha=\left\lfloor B S_{k}\right\rfloor+1}^{S_k}\frac{{N-S_{k}\choose M-S_{k}^\alpha}  {S_k\choose S_{k}^\alpha} }{{N\choose M}}
\label{eq:delta-Ek-hyper},
\end{equation}
i.e. the tail of  hypergeometric distribution, which estimated as follows 
\begin{equation}
\delta(E_k(B))\leq \frac{{N-S_{k}\choose M-\left\lfloor B S_{k}\right\rfloor-1}{S_k\choose \left\lfloor B S_{k}\right\rfloor+1}}{{N\choose M}}\frac{1-r^{S_{k}- \left\lfloor B S_{k}\right\rfloor}}{1-r}
\label{eq:delta-Ek-hyper-ub},
\end{equation}
where $r=\frac{\left[P-\frac{\left\lfloor B S_{k}\right\rfloor+1}{N}\right][1-B(k)]      }{\left[1-P-\frac{S_{k}}{N}+\frac{\left\lfloor B S_{k}\right\rfloor+1}{N}\right]\,B(k)}$, by the    \emph{Lemma} \ref{Lemma:hyper-tail-ub}. For $k=3$ the bounds in (\ref{eq:delta-Ek-binom-ub}) and (\ref{eq:delta-Ek-hyper-ub}) are plotted in the Figure \ref{fig:failure-prob}.

\section{Analysis of failures: Proof of Lemma \ref{Lemma:hyper-tail-ub}\label{appendix:proof}}
\begin{proof}
 Let us  consider the probability (\ref{eq:hyper-tail}) as follows
\setlength{\arraycolsep}{0.0em}
\begin{eqnarray}
&&\Prob\left(N_{\mu}^\alpha\geq \lfloor AN_{\mu}\rfloor+1\right)=\sum_{N_{\mu}^\alpha=\left\lfloor A\, N_{\mu}\right\rfloor+1}^{N_{\mu}}\Prob\left(N_{\mu}^\alpha\vert N_{\mu}\right)\nonumber\\
&&=\Prob\left(x_0\vert N_{\mu}\right)\sum_{N_{\mu}^\alpha=\left\lfloor A\, N_{\mu}\right\rfloor+1}^{N_{\mu}}\frac{\Prob\left(N_{\mu}^\alpha\vert N_{\mu}\right)}{\Prob\left(x_0\vert N_{\mu}\right)}\nonumber\\
&&=\Prob\left(x_0\vert N_{\mu}\right)\sum_{\ell=0}^{N_{\mu}-\left\lfloor A\, N_{\mu}\right\rfloor-1}\!\!\frac{\Prob\left(\left\lfloor A\, N_{\mu}\right\rfloor+1+\ell\vert N_{\mu}\right)}{\Prob\left(x_0\vert N_{\mu}\right)},
\end{eqnarray}
\setlength{\arraycolsep}{5pt}
where in above we  assumed that  the  $\sup_{x}\Prob\left(x\vert N_{\mu}\right)$ is  \emph{unique} on the interval $\left[\left\lfloor A\, N_{\mu}\right\rfloor+1,N_{\mu} \right]$ and  defined
\begin{equation}
\Prob\left(x_0\vert N_{\mu}\right)=\sup_{x\in\left[\left\lfloor A\, N_{\mu}\right\rfloor+1,N_{\mu} \right]}\Prob\left(x\vert N_{\mu}\right)\label{def:sup}.
\end{equation}
Now the mode of the probability distribution $\Prob\left(N_{\mu}^\alpha\vert N_{\mu}\right)$, for non-integer $ \frac{(N_{\mu}+1)(M+1)}{N+2}$, is located at $ \left\lfloor\frac{(N_{\mu}+1)(M+1)}{N+2}\right\rfloor$. We note that $ \left\lfloor\frac{(N_{\mu}+1)(M+1)}{N+2}\right\rfloor\leq\lfloor(N_{\mu}+1)(P+1/N)\rfloor\leq\lfloor(P+1/N)N_{\mu}\rfloor+1$ and hence for $P+1/N<B$ we have that $\Prob\left(x_0\vert N_{\mu}\right)=\Prob\left(\left\lfloor A\, N_{\mu}\right\rfloor+1\vert N_{\mu}\right)$. Thus we obtain 
\setlength{\arraycolsep}{0.0em}
\begin{eqnarray}
&&\Prob\left(N_{\mu}^\alpha\geq \lfloor AN_{\mu}\rfloor+1\right)=\Prob\left(\left\lfloor A\, N_{\mu}\right\rfloor+1\vert N_{\mu}\right)\nonumber\\
&&~~~~~~~~~~~~~~~\times\sum_{\ell=0}^{N_{\mu}-\left\lfloor A\, N_{\mu}\right\rfloor-1}\frac{\Prob\left(\left\lfloor A\, N_{\mu}\right\rfloor+1+\ell\vert N_{\mu}\right)}{\Prob\left(\left\lfloor A\, N_{\mu}\right\rfloor+1\vert N_{\mu}\right)}\label{eq:delta-k-hyper-sum}.
\end{eqnarray}
\setlength{\arraycolsep}{5pt}
Let us  for $\ell>0$ and $n=\left\lfloor A\, N_{\mu}\right\rfloor+1$ consider the ratio  
\begin{equation}
\frac{\Prob\left(n+\ell\vert N_{\mu}\right)}{\Prob\left(n\vert N_{\mu}\right)}
=\frac{{N-N_{\mu}\choose M-n-\ell}{N_{\mu}\choose n+\ell}
}{ {N-N_{\mu}\choose M-n}{N_{\mu}\choose n}  }.
\end{equation}
We note that 
\begin{equation}
\frac{{N\choose n}}{{N\choose n+\ell}}
=\frac{(N-n-\ell)!(n+\ell)!}{ (N-n)!n!}\geq \frac{n^\ell}{(N-n)^\ell}
\end{equation}
and 
\begin{equation}
\frac{{N\choose n}}{{N\choose n-\ell}}
=\frac{(N-n+\ell)!(n-\ell)!}{ (N-n)!n!}\geq \frac{(N-n)^\ell }{n^\ell},
\end{equation}
hence using above we obtain the following inequality  
\setlength{\arraycolsep}{0.0em}
\begin{eqnarray}
&&\frac{\Prob\left(n+\ell\vert N_{\mu}\right)}{\Prob\left(n\vert N_{\mu}\right)}\nonumber\\
&&~~~~~~~\leq\left\{\frac{(M-n)(N_{\mu}-n)      }{(N-N_{\mu}-M+n)n}\right\}^\ell\nonumber\\
%
%
&&~~~~~~~~~=\left\{\frac{\left[P-\frac{\left\lfloor A\, N_{\mu}\right\rfloor+1}{N}\right][1-A(\mu)]      }{\left[1-P-\frac{N_{\mu}}{N}+\frac{\left\lfloor A\, N_{\mu}\right\rfloor+1}{N}\right]\,A(\mu)}\right\}^\ell,
\end{eqnarray}
\setlength{\arraycolsep}{5pt}
where we defined $A(\mu)=\frac{\left\lfloor A\, N_{\mu}\right\rfloor+1}{N_{\mu}}$ and $P=M/N$. Finally, using above  in  (\ref{eq:delta-k-hyper-sum}) we obtain 
\setlength{\arraycolsep}{0.0em}
\begin{eqnarray}
&&\Prob\left(N_{\mu}^\alpha\geq \lfloor AN_{\mu}\rfloor+1\right)\nonumber\\
%
&&~~~~~~~~~\leq \Prob\left(\left\lfloor A\, N_{\mu}\right\rfloor+1\vert N_{\mu}\right)\sum_{\ell=0}^{N_{\mu}- \left\lfloor A\, N_{\mu}\right\rfloor-1}r^\ell\nonumber\\
&&~~~~~~~~~~~=\Prob\left(\left\lfloor A\, N_{\mu}\right\rfloor+1\vert N_{\mu}\right)\frac{1-r^{N_{\mu}- \left\lfloor A\, N_{\mu}\right\rfloor}}{1-r},
%
%
\end{eqnarray}
\setlength{\arraycolsep}{5pt}
where we defined 
\begin{equation}
r=\frac{\left[P-\frac{\left\lfloor A\, N_{\mu}\right\rfloor+1}{N}\right][1-A(\mu)]      }{\left[1-P-\frac{N_{\mu}}{N}+\frac{\left\lfloor A\, N_{\mu}\right\rfloor+1}{N}\right]\,A(\mu)}\label{def:r}.
\end{equation}
\end{proof}

\section{Analysis of failures: description  of  the algorithm\label{appendix:alg}}
Here we give details  of an algorithm which given the number of nodes $N$, assumed fraction of adversarial nodes parameter  $P$, a fraction of committee parameter $A$ and the probability of failure $\delta$ would give us  the \emph{maximum} number of committees $K$ and  committee sizes. The algorithm computes  the RHS of  (\ref{eq:delta-binom}) for the  
   $N=nK+r$ nodes where  the  $n$ and $n+1$ nodes are assigned, respectively, to the  $K-r$ and $r$ committees. Initially, all $N$ nodes are in one committee, and in subsequent iterations, the number of committees $K$ is increased by two until the probability  (\ref{eq:delta-binom}) is equal to  or less than $\delta$. We note that instead of (\ref{eq:delta-binom}) the upper bound (\ref{eq:delta-binom-ub}) can be used. The latter  has  lower numerical complexity but leads to slightly larger committee sizes. Furthermore, if  we assume  the hypergeometric distribution   (\ref{eq:prob-hyper})  then     the upper bound (\ref{eq:delta-hyper-ub}) can be used instead of (\ref{eq:delta-binom-ub}) in the algorithm. The pseudocode of the algorithm   is provided in the Algorithm \ref{alg:comm}. 
   
   \begin{algorithm}[tb]
\caption{The algorithm to compute minimal sizes of committees resilient to $\delta$ failure rate.}
\label{alg:comm}
\begin{algorithmic}[1]
\Require{$N$, $\delta$, $A$, $P$}
\Ensure{$K$, $n$, $r$, $\mathrm{Prob}$}
\State $K \gets 1$
\State $n \gets N$
\State $r \gets 0$
\State $m \gets 0$
\State $\mathrm{Prob} \gets 0$
\Repeat
\State //Save values of  $K$, $n$, $r$ and $\mathrm{Prob}$.
  \State $K_{-1} \gets K$
  \State $n_{-1} \gets n$
  \State $r_{-1} \gets r$
  \State $Prob_{-1} \gets Prob$
  \State //Compute next (odd) $K$.
  \State $m \gets m + 1$
  \State $K \gets 2m + 1$
  \State //Compute remainder, $r$, and quotient, $n$, when $N$ is divided by $K$.
  \State $r\gets rem(N,K)$
  \State $n\gets quot(N,K)$
  \If{$r > 0$}
  \State //Compute CDF of the Binomial(n,P) and Binomial(n+1,P).
    \State $\mathrm{Prob}_0 \gets \Prob(X \leq \lfloor An \rfloor \, \vert \, n, P)$
    \State $\mathrm{Prob}_1 \gets \Prob(X \leq \lfloor A(n+1) \rfloor \, \vert \, n+1, P)$
    \State //Compute probability of failure.
    \State $\Prob \gets 1 - \mathrm{Prob}_0^{K-r}\mathrm{Prob}_1^r$
  \Else
    \State $\mathrm{Prob}_0 \gets \Prob(X \leq \lfloor An \rfloor \, \vert \, n, P)$
    \State $Prob \gets 1 - \mathrm{Prob}_0^K$
  \EndIf
\Until{$\mathrm{Prob} > \delta$}
\State $K \gets K_{-1}$
\State $n \gets n_{-1}$
\State $r \gets r_{-1}$
\State $\mathrm{Prob} \gets \mathrm{Prob}_{-1}$
\end{algorithmic}
\end{algorithm}

\section{Analysis of failures:  properties of  $\delta(E_k(A))$ \label{appendix:delta-prop}}
In this section we consider  properties of the failure probability  $\delta(E_k(A))$ which hold  
for \emph{any} probability distribution $\Prob\left(N_{1}^\alpha,\ldots, N_{K}^\alpha\vert N_{1},\ldots,N_{K}\right)$. First, we show that the probability  $\delta(E_k(A))$, which is defined for $k=1$ in  (\ref{def:delta-E1}), for $k=2$ in  (\ref{def:delta-E2}) and for $k\geq3$ in (\ref{def:delta-Ek}),  is a monotonic non-increasing function of $A$. 
\begin{prop}
For the $A,B\in[0,1]$ such that $A\leq B$ and $k\geq1$ the probability of failure  $\delta(E_k(B)\leq \delta(E_k(A)$.
\label{Proposition:E_1(A)-E_1(B)-ineq}
 \end{prop}
 \begin{proof}
 Let  us consider the $k=1$ case as follows  
\setlength{\arraycolsep}{0.0em}
\begin{eqnarray}
&&1-\delta(E_1(A))\nonumber\\
&&~~~=\Prob\left(N_{1}^\alpha\leq\lfloor A N_{1}\rfloor ,\ldots, N_{K}^\alpha\leq\lfloor A N_{K}\rfloor\right)\nonumber\\
%
%
%
&&~~~=\sum_{N_{1}^\alpha=0}^{\lfloor AN_{1}\rfloor}\cdots \sum_{N_{K}^\alpha=0}^{\lfloor AN_{K}\rfloor}\Prob\left(N_{1}^\alpha,\ldots, N_{K}^\alpha\vert N_{1},\ldots,N_{K}; M\right)\nonumber\\
&&~~~~~\leq
\sum_{N_{1}^\alpha=0}^{\lfloor BN_{1}\rfloor}\cdots \sum_{N_{K}^\alpha=0}^{\lfloor BN_{K}\rfloor}\Prob\left(N_{1}^\alpha,\ldots, N_{K}^\alpha\vert N_{1},\ldots,N_{K}; M\right)\nonumber\\
&&~~~~~~~~~~=1-\delta(E_1(B))
\label{eq:delta-E1(A)-E1(B)-ineq}
\end{eqnarray}
\setlength{\arraycolsep}{5pt}
%
and hence $\delta(E_1(B))\leq\delta(E_1(A))$. We note that the same argument can be used to show that $\delta(E_k(B))\leq\delta(E_k(A))$ for $k\geq2$. 
 \end{proof}

 Second, for a fixed $A$ we establish how the $\delta(E_k(A))$ of different $k$ are related to which other in the following two propositions. 

\begin{prop}
For $A\in[0,1]$ the probability of failure $ \delta(E_2(A)\leq \delta(E_1(A)$. 
\label{Proposition:E_2-E_1-ineq}
 \end{prop}

\begin{proof}
First,  we consider
\setlength{\arraycolsep}{0.0em}
\begin{eqnarray}
&&1-\delta(E_2(A))\nonumber\\
&&~~~~~=\Prob\left(S_{2}^\alpha\leq\lfloor A S_{2}\rfloor ,\ldots, S_{\tilde{K}}^\alpha\leq\lfloor A S_{\tilde{K}}\rfloor\right)
\nonumber\\
&&~~~~~~~\geq 
\Prob\left(S_{1}^\alpha\leq\lfloor A S_{1}\rfloor ,\ldots, S_{\tilde{K}}^\alpha\leq\lfloor A S_{\tilde{K}}\rfloor\right)\nonumber\\
&&~~~~~=\sum_{N_{1}^\alpha}\cdots \sum_{N_{K}^\alpha}\Prob\left(N_{1}^\alpha,\ldots, N_{K}^\alpha\vert N_{1},\ldots,N_{K}; M\right)\nonumber\\
&&~~~~~\times\Ind\left[N^{\alpha}_{1}\leq \lfloor AN_{1}\rfloor\right]\nonumber\\
&&~~~~~\times\prod_{\mu=2}^{\tilde{K}}\Ind\left[N^{\alpha}_{2\mu}+N^{\alpha}_{2\mu+1}\leq \lfloor A(N_{2\mu}+N_{2\mu+1})\rfloor\right]\label{eq:delta-E2-ineq},
\end{eqnarray}
\setlength{\arraycolsep}{5pt}
where in above we used the definition  (\ref{eq:delta-E2}). Second, we show that  
\setlength{\arraycolsep}{0.0em}
\begin{eqnarray}
&&\Ind\left[N^{\alpha}_{2\mu}+N^{\alpha}_{2\mu+1}\leq
\lfloor A(N_{2\mu}+N_{2\mu+1})\rfloor\right]\nonumber\\
&&~~~~~\geq \Ind\left[N^{\alpha}_{2\mu}\leq
\lfloor AN_{2\mu}\rfloor\right]
\Ind\left[N^{\alpha}_{2\mu+1}\leq
\lfloor AN_{2\mu+1}\rfloor\right]\label{eq:1-ineq}
\end{eqnarray}
\setlength{\arraycolsep}{5pt}
Here we only need to prove that  when the LHS in above is $0$ then the  RHS  can not be $1$. Let us assume that the latter  is not true then $N^{\alpha}_{2\mu}+N^{\alpha}_{2\mu+1}>
\lfloor A(N_{2\mu}+N_{2\mu+1})\rfloor$, $N^{\alpha}_{2\mu}\leq
\lfloor AN_{2\mu} \rfloor$ and $N^{\alpha}_{2\mu+1}\leq
\lfloor AN_{2\mu+1}\rfloor$. However, the last two inequalities imply that  $N^{\alpha}_{2\mu}+N^{\alpha}_{2\mu+1}\leq
\lfloor AN_{2\mu} \rfloor+\lfloor AN_{2\mu+1}\rfloor$, but $\lfloor AN_{2\mu} \rfloor+\lfloor AN_{2\mu+1}\rfloor\leq  \lfloor A(N_{2\mu}+N_{2\mu+1})\rfloor$ and hence $$N^{\alpha}_{2\mu}+N^{\alpha}_{2\mu+1}\leq  \lfloor A(N_{2\mu}+N_{2\mu+1})\rfloor <N^{\alpha}_{2\mu}+N^{\alpha}_{2\mu+1}$$ which is not possible.  Thus the inequality (\ref{eq:1-ineq}) is true. Now using the latter in   (\ref{eq:delta-E2-ineq}) gives us 
\setlength{\arraycolsep}{0.0em}
\begin{eqnarray}
&&1-\delta(E_2(A))\nonumber\\
%
%
&&~~~~~\geq\sum_{N_{1}^\alpha}\cdots \sum_{N_{K}^\alpha}\Prob\left(N_{1}^\alpha,\ldots, N_{K}^\alpha\vert N_{1},\ldots,N_{K}; M\right)\nonumber\\
&&~~~~~~~\times\Ind\left[N^{\alpha}_{1}\leq \lfloor AN_{1}\rfloor\right]\prod_{\mu=2}^{\tilde{K}}
\Ind\left[N^{\alpha}_{2\mu}\leq
\lfloor AN_{2\mu}\rfloor\right]\nonumber\\
&&~~~~~~~\times
\Ind\left[N^{\alpha}_{2\mu+1}\leq
\lfloor AN_{2\mu+1}\rfloor\right]\nonumber\\
&&~~~~~~~~~=\Prob\left(N_{1}^\alpha\leq\lfloor A N_{1}\rfloor ,\ldots, N_{K}^\alpha\leq\lfloor A N_{K}\rfloor\right)\nonumber\\
&&~~~~~~~~~=1-\delta(E_1(A))\label{eq:delta-E2-E1-ineq}
\end{eqnarray}
\setlength{\arraycolsep}{5pt}
and hence $\delta(E_2(A))\leq\delta(E_1(A))$. 
\end{proof}

\begin{prop}
 For $A\in[0,1]$ and $k\geq3$ the probability of failure  $ \delta(E_k(A)\leq \delta(E_1(A)$.  
\label{Proposition:E_k-E_1-ineq}
 \end{prop}

\begin{proof}
First, from  the definition (\ref{def:delta-Ek}) follows the equality 
\setlength{\arraycolsep}{0.0em}
\begin{eqnarray}
1-\delta(E_k(A))&=&\Prob
\left(\sum_{\nu=1}^kN_{\nu}^\alpha\leq \lfloor A \sum_{\nu=1}^kN_{\nu}\rfloor\right)\nonumber\\
&=&\sum_{N_{1}^\alpha}\cdots \sum_{N_{K}^\alpha}\Prob\left(N_{1}^\alpha,\ldots, N_{K}^\alpha\vert N_{1},\ldots,N_{K}; M\right)\nonumber\\
&&~~~~~\times\Ind\left[\sum_{\nu=1}^kN_{\nu}^\alpha\leq \lfloor A \sum_{\nu=1}^kN_{\nu}\rfloor\right].
\label{eq:1-delta-E_k}
\end{eqnarray}
\setlength{\arraycolsep}{5pt}
Second, in order to show that 
\setlength{\arraycolsep}{0.0em}
\begin{eqnarray}
\Ind\left[\sum_{\nu=1}^kN_{\nu}^\alpha\leq \lfloor A \sum_{\nu=1}^kN_{\nu}\rfloor\right]
&\geq &\prod_{\nu=1}^k\Ind\left[N^{\alpha}_{\nu}\leq
\lfloor AN_{\nu}\rfloor\right]\label{eq:1-k-ineq}
\end{eqnarray}
\setlength{\arraycolsep}{5pt}
we assume that  $ \sum_{\nu=1}^kN_{\nu}^\alpha> \lfloor A \sum_{\nu=1}^kN_{\nu}\rfloor$ and $N^{\alpha}_{\nu}\leq
\lfloor AN_{\nu}\rfloor $ for all $\nu\in[k]$.  This implies that $\sum_{\nu=1}^kN^{\alpha}_{\nu}\leq
\sum_{\nu=1}^k\lfloor AN_{\nu}\rfloor$, but $\sum_{\nu=1}^k\lfloor AN_{\nu}\rfloor\leq \lfloor A \sum_{\nu=1}^kN_{\nu}\rfloor$ and $$\sum_{\nu=1}^kN^{\alpha}_{\nu}\leq \lfloor A \sum_{\nu=1}^kN_{\nu}\rfloor <\sum_{\nu=1}^kN_{\nu}^\alpha$$ which can not be true. Hence   (\ref{eq:1-k-ineq}) is correct and using this inequality in (\ref{eq:1-delta-E_k}) gives us 
\setlength{\arraycolsep}{0.0em}
\begin{eqnarray}
1-\delta(E_k(A))&\geq&\sum_{N_{1}^\alpha}\cdots \sum_{N_{K}^\alpha}\Prob\left(N_{1}^\alpha,\ldots, N_{K}^\alpha\vert N_{1},\ldots,N_{K}; M\right)\nonumber\\
&&~~~~~\times\prod_{\nu=1}^k\Ind\left[N^{\alpha}_{\nu}\leq
\lfloor AN_{\nu}\rfloor\right]\nonumber\\
&&\geq\sum_{N_{1}^\alpha}\cdots \sum_{N_{K}^\alpha}\Prob\left(N_{1}^\alpha,\ldots, N_{K}^\alpha\vert N_{1},\ldots,N_{K}; M\right)\nonumber\\
&&~~~~~\times\prod_{\mu=1}^K\Ind\left[N^{\alpha}_{\mu}\leq
\lfloor AN_{\mu}\rfloor\right]\nonumber\\
&&~~~=1-\delta(E_1(A))\label{eq:1-delta-E_k-ineq}
\end{eqnarray}
\setlength{\arraycolsep}{5pt}
from which follows $ \delta(E_k(A)\leq \delta(E_1(A)$.  
 \end{proof}



\end{document}